\documentclass[prl,aps,superscriptaddress,nopacs,twocolumn,longbibliography]{revtex4-1}

\usepackage{graphicx,epsfig,amsmath,amssymb,verbatim,color}
\usepackage{dsfont}
\usepackage{float}
\usepackage{tikz}
\usepackage{pgfplots}
\usepackage{xcolor}
\definecolor{darkred}{rgb}{0.8,0.1,0.1}
\usepackage{soul}
\usepackage{amsmath,mathtools}
\usepackage{enumerate}
\usepackage[T1]{fontenc}
\usepackage[utf8]{inputenc}
\usepackage{pifont}
\usepackage{hyperref}
\hypersetup{colorlinks=true,citecolor=blue,linkcolor=blue,filecolor=blue,urlcolor=blue,breaklinks=true}

\usepackage{colortbl}
\usepackage{pifont}
\definecolor{Gray}{gray}{0.92}
\definecolor{Gray2}{gray}{0.75}
\definecolor{maroon}{cmyk}{0,0.87,0.68,0.32}
\usepackage{booktabs}
\usepackage{makecell}
\usepackage{diagbox}
\usepackage{multirow}

\newtheorem{definition}{Definition}
\newtheorem{proposition}{Proposition}
\newtheorem{lemma}[proposition]{Lemma}

\newtheorem{theorem}[proposition]{Theorem}
\newtheorem{remark}{Remark}
\newtheorem{corollary}[proposition]{Corollary}

\newenvironment{proof}{\noindent \textit{{Proof.}~}}{\hfill $\square$}

\def\squareforqed{\hbox{\rlap{$\sqcap$}$\sqcup$}}
\def\qed{\ifmmode\squareforqed\else{\unskip\nobreak\hfil
\penalty50\hskip1em\null\nobreak\hfil\squareforqed
\parfillskip=0pt\finalhyphendemerits=0\endgraf}\fi}
\def\endenv{\ifmmode\;\else{\unskip\nobreak\hfil
\penalty50\hskip1em\null\nobreak\hfil\;
\parfillskip=0pt\finalhyphendemerits=0\endgraf}\fi}

\newcommand{\nc}{\newcommand}
\nc{\rnc}{\renewcommand}
\nc{\bra}[1]{\langle#1|}
\nc{\ket}[1]{|#1\rangle}
\nc{\<}{\langle}
\rnc{\>}{\rangle}
\nc{\ketbra}[2]{|#1\rangle\!\langle#2|}
\nc{\braket}[2]{\langle#1|#2\rangle}
\nc{\braandket}[3]{\langle #1|#2|#3\rangle}
\nc{\proj}[1]{| #1\rangle\!\langle #1 |}
\nc{\avg}[1]{\langle#1\rangle}
\nc{\Rank}{\operatorname{Rank}}
\nc{\smfrac}[2]{\mbox{$\frac{#1}{#2}$}}
\nc{\tr}{\operatorname{Tr}}
\nc{\ox}{\otimes}

\nc{\cA}{{\cal A}}
\nc{\cB}{{\cal B}}
\nc{\cC}{{\cal C}}
\nc{\cD}{{\cal D}}
\nc{\cE}{{\cal E}}
\nc{\cF}{{\cal F}}
\nc{\cG}{{\cal G}}
\nc{\cH}{{\cal H}}
\nc{\cI}{{\cal I}}
\nc{\cJ}{{\cal J}}
\nc{\cK}{{\cal K}}
\nc{\cL}{{\cal L}}
\nc{\cM}{{\cal M}}
\nc{\cN}{{\cal N}}
\nc{\cO}{{\cal O}}
\nc{\cP}{{\cal P}}
\nc{\cQ}{{\cal Q}}
\nc{\cR}{{\cal R}}
\nc{\cS}{{\cal S}}
\nc{\cT}{{\cal T}}
\nc{\cV}{{\cal V}}
\nc{\cU}{{\cal U}}
\nc{\cX}{{\cal X}}
\nc{\cY}{{\cal Y}}
\nc{\cZ}{{\cal Z}}
\nc{\cW}{{\cal W}}

\nc{\RR}{{{\mathbb R}}}
\nc{\CC}{{{\mathbb C}}}
\nc{\FF}{{{\mathbb F}}}
\nc{\NN}{{{\mathbb N}}}
\nc{\ZZ}{{{\mathbb Z}}}
\nc{\QQ}{{{\mathbb Q}}}
\nc{\UU}{{{\mathbb U}}}
\nc{\EE}{{{\mathbb E}}}
\nc{\id}{{\operatorname{id}}}

\nc{\1}{{\mathds{1}}}
\nc{\supp}{{\operatorname{supp}}}

\def\ve{\varepsilon}

\usepackage{mathrsfs}
\nc{\plus}{{\scalebox{0.7}{\rm +}}}
\nc{\PSD}{\HERM_{\plus}}
\nc{\PD}{\HERM_{\plus\plus}}
\nc{\polarPSD}[1]{{#1}_{\plus}^{\circ}}
\nc{\polarPSDre}[1]{{#1}_{\plus}^{\star}}
\nc{\polarPD}[1]{{#1}_{\plus\plus}^{\circ}}
\nc{\HERM}{\mathscr{H}}
\nc{\cvxset}{\mathscr{C}}
\nc{\density}{\mathscr{D}}
\nc{\subdensity}{\mathscr{D}_\bullet}
\nc{\bcC}{\cC}
\nc{\Meas}{{\scriptscriptstyle \rm M}}
\nc{\Proj}{{{\scriptscriptstyle \rm P}}}
\nc{\RM}{{{\mathscr{R}}}}
\nc{\sK}{{{\mathscr{K}}}}
\nc{\sS}{{{\mathscr{S}}}}
\nc{\sT}{{{\mathscr{T}}}}
\nc{\sA}{{{\mathscr{A}}}}
\nc{\sB}{{{\mathscr{B}}}}
\nc{\sC}{{{\mathscr{C}}}}
\nc{\sE}{{{\mathscr{E}}}}
\nc{\sL}{{{\mathscr{L}}}}
\nc{\sG}{{{\mathscr{G}}}}
\nc{\sF}{{{\mathscr{F}}}}
\nc{\sI}{{{\mathscr{I}}}}
\nc{\sN}{{{\mathscr{N}}}}
\nc{\sM}{{{\mathscr{M}}}}
\nc{\END}{\operatorname{End}}
\nc{\PERM}{\mathfrak{\sigma}}

\nc{\Cone}{\text{\rm Cone}}
\nc{\sep}{{\SEP}}
\nc{\BS}{{\scriptscriptstyle \rm {BS}}}
\nc{\Sand}{{\scriptscriptstyle  \rm S}}
\nc{\Petz}{{\scriptscriptstyle  \rm P}}
\nc{\Hypo}{{\scriptscriptstyle  \rm H}}
\nc{\DD}{{{\mathbb D}}}
\nc{\suchthat}{\text{\rm s.t.}}
\nc{\PPT}{\text{\rm PPT}}
\nc{\Rains}{\text{\rm Rains}}
\nc{\WD}{\text{\rm WD}}
\nc{\new}{\text{\rm new}}
\nc{\sfT}{\mathsf T}
\nc{\SEP}{\text{\rm SEP}}
\nc{\PSEP}{\text{\rm PSEP}}
\nc{\CPTP}{\text{\rm CPTP}}
\nc{\POVM}{\text{\rm POVM}}
\nc{\PVM}{\text{\rm PVM}}
\nc{\CP}{\text{\rm CP}}
\nc{\adv}{\text{\rm adv}}
\nc{\spec}{\text{\rm spec}}
\nc{\poly}{\text{\rm poly}}
\nc{\End}{\operatorname{End}}
\nc{\Par}{\operatorname{Par}}
\nc{\RNG}{\operatorname{RNG}}
\nc{\epi}{\boldsymbol{\operatorname{epi}}}
\nc{\op}{\boldsymbol{\operatorname{op}}}
\newcommand{\reg}{\infty}
\newcommand{\amo}{\text{\rm amo}}
\newcommand{\Renyi}{R\'{e}nyi }
\newcommand{\revh}{\underline{h}}

\begin{document}
\title{Adversarial quantum channel discrimination}

\author{Kun Fang}
\email{kunfang@cuhk.edu.cn}
\affiliation{School of Data Science, The Chinese University of Hong Kong, Shenzhen, Guangdong, 518172, China}

\author{Hamza Fawzi}
\email{h.fawzi@damtp.cam.ac.uk}
\affiliation{Department of Applied Mathematics and Theoretical Physics,  University of Cambridge, Cambridge CB3 0WA, United Kingdom}

\author{Omar Fawzi}
\email{omar.fawzi@ens-lyon.fr}
\affiliation{Inria, ENS de Lyon, UCBL, LIP, 69342, Lyon Cedex 07, France}

\date{\today}

\begin{abstract}
{We introduce a new framework for quantum channel discrimination in an adversarial setting, where the tester plays against an adversary. We show that in asymmetric hypothesis testing, the optimal type-II error exponent is precisely characterized by a new notion of quantum channel divergence (termed the minimum output channel divergence). This serves as a direct analog of the quantum Stein's lemma in this new framework, and complements previous studies on ``best-case'' channel discrimination, thereby providing a complete understanding of the ultimate limits of quantum channel discrimination.
Notably, the optimal error exponent can be achieved by simple non-adaptive adversarial strategies, and—despite the need for regularization—it remains efficiently computable and satisfies the strong converse property in general. Furthermore, we show that entropy accumulation, a powerful tool in quantum cryptography, can be reframed as an adversarial channel discrimination problem, establishing a new connection between quantum information theory and quantum cryptography.}
\end{abstract}

\maketitle

A quantum channel describes the most general physical transformation that quantum states can undergo. It captures not only idealized processes like unitary evolutions but also non-ideal, noisy, or dissipative processes, making them crucial for understanding real-world quantum systems. The task of quantum channel discrimination is to identify which member of a given set of quantum channels governs the physical process in a black-box scenario.  It is pivotal to various quantum information tasks and gives insights into a wide range of quantum protocols and applications~\cite{dariano_using_2001,chiribella_memory_2008,hayashi_discrimination_2009,harrow_adaptive_2010,wilde2020amortized,Cooney2016,wang2019resource,pirandola_fundamental_2019,WW2019,bergh_parallelization_2024}, from quantum foundations (e.g., exploring the quantum advantage of entanglement~\cite{piani2009all,takagi_operational_2019,bae_more_2019,skrzypczyk_robustness_2019}) to quantum communication (e.g., estimating quantum channel capacity~\cite{Wang2012,datta_smooth_2013,Wang2019,wang2019converse,fang2021geometric,fang2025towards}), quantum sensing (e.g., quantum reading and quantum illumination~\cite{pirandola_advances_2018,zhuang_ultimate_2020}) and even quantum biology~\cite{spedalieri_detecting_2020,pereira2020quantum}. 

{While previous studies on quantum channel discrimination have focused exclusively on the conventional ``best-case'' scenario~\cite{acin2001statistical,duan2007entanglement,chiribella_memory_2008,piani2009all,duan2009perfect,bae_more_2019,takagi_operational_2019,skrzypczyk_robustness_2019,fang2020chain,zhuang_ultimate_2020,bavaresco2021strict,debry2023experimental,sugiura2024power}, where the tester has full control over input states, intermediate operations, and final measurements, this work introduces a ``worst-case'' adversarial framework. This scenario is motivated by applications such as quantum device verification~\cite{zhu2019efficient}, where a state preparation device from an untrusted manufacturer is expected to produce a specific resource state (e.g., a Bell state, magic state, or coherent state), but may instead output arbitrary junk states if faulty or maliciously designed. 
In the adversarial setting, the roles of the tester and adversary are fundamentally different, resulting in a competitive game: the tester designs the measurement, while the adversary controls the choice of input states and potentially intermediate updates. 
This leads to a highly non-trivial minimax optimization problem, where the tester must design a universal, state-agnostic measurement strategy that  suppresses the worst-case error induced by any possible adversarial strategy.}

Technically, the discrimination task is formulated as a hypothesis testing problem, with the goal of finding discrimination strategies that give the optimal trade-off between two kinds of error probabilities, namely the probabilities of false detection (type-I error) and false rejection (type-II error)~\cite{hiai1991proper,Ogawa2000}. The challenge in adversarial discrimination arises from the potential capabilities of the adversary, who may access the channel's environmental system—the ancillary quantum system that interacts with the main system and can cause information leakage—or possess (potentially unbounded) quantum memory that stores partial information from previous rounds and uses it in subsequent rounds to enable adaptive strategies. This raises a central question: how effectively can the tester distinguish between two quantum channels while playing against such an adversary?

This work provides a complete answer to the above question in the context of asymmetric hypothesis testing. Specifically, we show that in adversarial quantum channel discrimination, the optimal type-II error decays exponentially at a rate characterized by a new notion of quantum channel divergence, termed the minimum output channel divergence, provided the type-I error remains within a fixed threshold. 
This result serves as a direct analog of the quantum Stein's lemma in adversarial channel discrimination. 
{It recovers the renowned quantum Stein's lemma for states~\cite{hiai1991proper,Ogawa2000} and complements previous studies on best-case channel discrimination~\cite{wilde2020amortized,WW2019,fang2020chain}, thereby providing a complete understanding of the ultimate limits of quantum channel discrimination.}
Notably, the optimal error exponent can be achieved via simple non-adaptive strategies by the adversary, {and despite the need for regularization, it remains efficiently
computable and satisfies the strong converse property in general. 
These features are unique in quantum information theory.  We also show that the entropy accumulation~\cite{arnon2018practical,liu2018device,george2022finite}, a powerful tool in quantum cryptography, can be reframed as an adversarial quantum channel discrimination problem, establishing a new connnection between quantum information theory and quantum cryptography and  providing a solution to the dual formulation of the open problem in~\cite{metger2022generalised}.}

\smallskip

\smallskip
\paragraph{Adversarial quantum channel discrimination.---}  Consider a scenario where a tester interacts with an untrusted quantum device that generates quantum states upon request. The device guarantees that the states are produced by either a quantum channel $\cN$ or a quantum channel $\cM$~\footnote{Adversarial state discrimination was previously studied in the classical case by~\cite{brandao2020adversarial}. In~\cite{brandao2020adversarial}, the problem is formulated in terms of convex sets $P, Q$ of probability distributions. For the quantum generalization (which is operationally well-motivated), we define the sets $P$ and $Q$ as the images of all input probability distributions of certain physical channels, denoted $\cN$ and $\cM$. This means that the channel $\cN$ takes as input a description of some $p \in P$ and outputs a sample from $p$. If $P$ is the convex hull of a finite set of probability distributions $(p_1, \dots, p_k)$, such a map $\cN$ can be constructed by letting $\cN(\ketbra{i}{j}) = \delta_{i=j} p_i$. For more general convex sets, the channel $\cN$ would require an infinite-dimensional input, which is beyond the scope of this paper.}. The tester is allowed to request multiple samples from the device and perform measurements to distinguish between the two cases.
{In this work, we consider the ``most extreme'' scenario, where the adversary has access to the environmental system of the channel and possesses full knowledge of the true channel. This operationally reflects the situation in which the tester plays against an adversary who manufactures the device. It also ensures that our results establish a fundamental limit that applies to all other possible variants.}

{More formally, let $\cN_{A\to B}$ and $\cM_{A\to B}$ be the two quantum channels to be distinguished, and let $\cU_{A\to BE}$ and $\cV_{A\to BE}$ denote their respective Stinespring dilations, with $E$ representing the environmental system. Let $\CPTP(X\!:\!Y)$ denote the set of completely positive and trace-preserving maps from input system $X$ to output system $Y$. An \emph{adaptive strategy} for adversarial discrimination proceeds as follows (see Figure~\ref{fig: adversarial setting}(a)). Suppose the device operates as channel $\cN$. Initially, the adversary prepares a quantum state via an operation $\cP^1 \in \CPTP(R_0E_0\!:\!A_1R_1)$, where $R_0$ and $E_0$ are trivial ($|R_0| = |E_0| = 1$), and sends system $A_1$ through the channel $\cU$, generating the output state $\cU \circ \cP^1$ and returning system $B_1$ to the tester. In the next round, the adversary performs an internal update $\cP^2 \in \CPTP(E_1R_1\!:\!A_2R_2)$, utilizing information stored in the quantum memory $R_1$ and the environmental system $E_1$ from the previous round. The adversary then sends system $A_2$ through the channel $\cU$ again, producing the output state $\cU \circ \cP^2 \circ \cU \circ \cP^1$ and returning system $B_2$ to the tester. This process can be repeated for $n$ rounds. A \emph{non-adaptive strategies} for adversarial discrimination is a subclass of adaptive strategies that disregards the environmental systems $E_i$ and performs no updates between rounds (see Figure~\ref{fig: adversarial setting}(b)), that is, taking the operations $\cP_i, \cQ_i$ ($i \geq 2$) simply as identity maps, with the choice $R_i = A_{i+1} \cdots A_n$. }

{After $n$ rounds of state generation, the tester obtains an overall state on $B_1\cdots B_n$ in their possession as:
\begin{align*}
\rho[\{\cP^i\}_{i=1}^n] := \tr_{R_nE_n} \prod_{i=1}^{n} \left[\cU_{A_i \to B_i E_i} \circ \cP^i_{R_{i-1}E_{i-1} \to A_{i}R_{i}}\right].
\end{align*}
Similarly, if the device is governed by $\cM$ and the internal operations by the adversary are given by $\cQ^i$, then the overall state is given by 
\begin{align*}
\sigma[\{\cQ^i\}_{i=1}^n] := \tr_{R_nE_n} \prod_{i=1}^{n} \left[\cV_{A_i \to B_i E_i} \circ \cQ^i_{R_{i-1}E_{i-1} \to A_{i}R_{i}}\right].  
\end{align*}
The tester needs to perform a binary quantum measurement $\{M_n, I-M_n\}$ on systems $B_1 \cdots B_n$ to determine which channel was used inside the black box.}

\begin{figure}[h]
    \centering
    \includegraphics[width=0.46\textwidth]{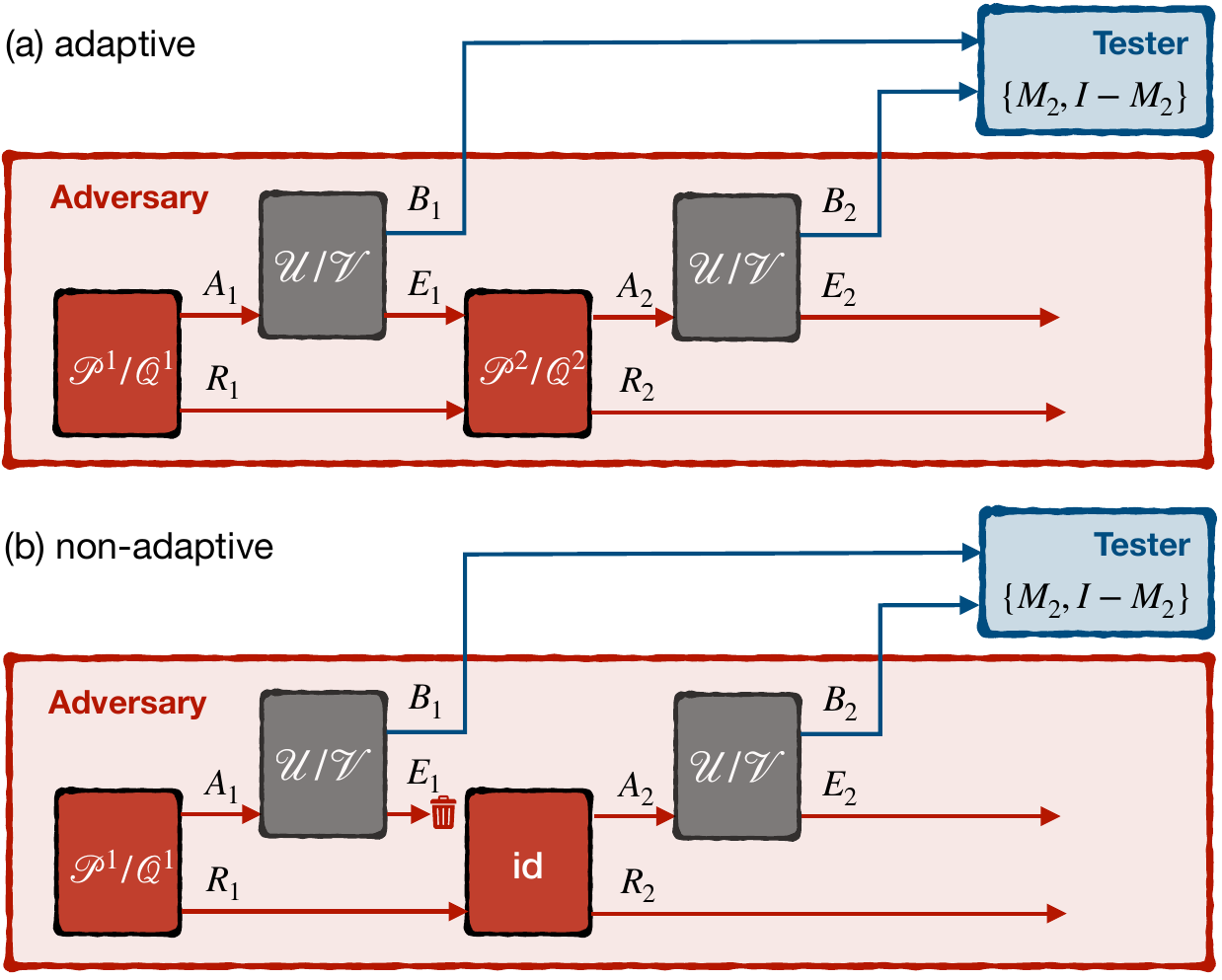}
        \caption{{Adaptive and non-adaptive strategies for adversarial quantum channel discrimination. Here, $\cU$ and $\cV$ (in gray) are the Stinespring dilations of the quantum channels $\cN$ and $\cM$, respectively; $\cP^i$ and $\cQ^i$ (in red) denote the adversary's internal operations; $\id$ is the identity map; and $\{M_2, I-M_2\}$ (in blue) represents the tester's quantum measurement. (a) Adaptive strategies: the adversary accesses the environmental systems $E_i$ and uses quantum memory $R_i$ to implement adaptive strategies. (b) Non-adaptive strategies: the adversary ignores the environmental systems $E_i$ and performs no updates (i.e., identity map) between rounds.}}
    \label{fig: adversarial setting}
\end{figure}

{Due to limited knowledge of the device's internal workings and the adversary's strategies, the tester only has access to partial information, knowing that their possessed state belongs to one of two sets:
$\sA_n := \{\rho[\{\cP^i\}_{i=1}^n] : \cP^i \in \CPTP(R_{i-1}E_{i-1}\!\!:\!\!A_iR_i), \forall R_i, \forall i\}$
or
$\sB_n := \{\sigma[\{\cQ^i\}_{i=1}^n] : \cQ^i \in \CPTP(R_{i-1}E_{i-1}\!\!:\!\!A_iR_i), \forall R_i, \forall i\}$
where the adversary's internal memory $R_i$ may have arbitrarily large dimension. This makes the discrimination problem particularly challenging, as the tester must design a state-agnostic measurement strategy that suppresses the worst-case error regardless of the adversary's choices. To capture this worst-case scenario, we define the type-I error as $\alpha(\sA_n, M_n) := \sup_{\rho_n \in \sA_n} \tr[\rho_n (I-M_n)]$ and the type-II error as $\beta(\sB_n, M_n) := \sup_{\sigma_n \in \sB_n} \tr[\sigma_n M_n]$.}

In asymmetric hypothesis testing, the goal is to analyze the optimal decay rate (also known as the Stein's exponent) of the type-II error probability when the type-I error is constrained to be below a fixed threshold $\ve$. Specifically, we consider the asymptotic behavior of
\begin{align*}
    \beta_{n, \ve}(\cN\|\cM) := \inf_{0\leq M_n\leq I}\left\{\beta(\sB_n, M_n): \alpha(\sA_n, M_n) \leq \ve\right\},
\end{align*}
as $n$ becomes large. {This leads to a minimax optimization problem that is generally difficult to analyze. Furthermore, because the adversary's strategies can be adaptive, the sets $\sA_n$ and $\sB_n$ are not permutation invariant, so existing results such as the quantum asymptotic equipartition property~\cite{fang2024generalized} do not directly apply.}

\smallskip
{\paragraph{Adversarial quantum Stein's lemma.---} We show that the optimal error exponent in adversarial channel discrimination is exactly characterized by a new notion of quantum channel divergence, termed the \emph{minimum output channel divergence}:
\begin{align*}
D^{\inf}(\cN\|\cM) := \inf \left\{ D(\cN(\rho)\| \cM(\sigma)): \rho, \sigma \in \density(A)\right\},
\end{align*}
where $\density(A)$ denotes the set of density operators on system $A$, and $D(\rho\|\sigma):= \tr[\rho(\log \rho - \log \sigma)]$ is the Umegaki relative entropy~\cite{umegaki1962conditional}. 
In contrast to best-case channel discrimination, where the optimal exponent is given by~\cite{WW2019,fang2020chain}, $\sup \{D(\cN_{A\to B}(\rho_{RA})\|\cM_{A\to B}(\rho_{RA})): \rho \in \density(RA)\}$, our new divergence independently optimizes over the input states for each channel and does not require a shared reference system $R$. This distinction turns out to be essential: as shown in the Supplemental Material~\cite{fang2025sm}, restricting to a common input state—structurally closer to the best-case divergence—fails to satisfy the chain rule, which is crucial for analyzing adaptive strategies and concluding our main results.}

\begin{theorem}(Adversarial quantum Stein's lemma.)\label{thm: adversarial quantum channel Steins lemma}
Let $\cN \in \CPTP(A\!:\!B)$ and $\cM \in \CP(A\!:\!B)$ (the set of completely positive maps). 
Then for any $\ve \in (0,1)$, 
\begin{align}\label{eq: Adversarial quantum Stein's lemma}
    \lim_{n\to \infty} - \frac{1}{n} \log \beta_{n, \ve}(\cN\|\cM) = D^{\inf,\reg}(\cN\|\cM),
\end{align}
where $D^{\inf,\reg}(\cN\|\cM) := \lim_{n \to \infty} D^{\inf}(\cN^{\ox n}\|\cM^{\ox n})/n$ can be efficiently computed via semidefinite programs. Moreover, the Stein's exponent can be achieved by non-adaptive strategies, indicating that adaptive strategies by the adversary offer no advantage over non-adaptive ones.
\end{theorem}

{The minimum output channel divergence $D^{\inf}(\cN\|\cM)$ vanishes if and only if there exist input states $\rho$ and $\sigma$ such that the output states $\cN(\rho)$ and $\cM(\sigma)$ are identical. While this can occur in certain
special cases, it is not generic. For example, in quantum device verification, the channel $\cN$ is designed to prepare a resource state, whereas the channel $\cM$ produces a junk state; by definition, these outputs cannot coincide.
Moreover, the minimum output channel divergence is generally non-additive~\cite{fang2025efficient}, which necessitates taking the regularized limit in the Stein's exponent $D^{\inf,\reg}(\cN\|\cM)$. While regularization typically complicates estimation, the fact that the optimal exponent is achieved by non-adaptive strategies leads to a significant simplification: it admits explicit convergence bounds, even in the presence of regularization. This exhibits a unique feature in quantum information theory and enables efficient computation of the exponent via a polynomial-time algorithm~\cite{fang2025sm}. In sharp contrast, for best-case channel discrimination, the optimal exponent is not even known to be computable by any means, regardless of computational resources.}

{The strong converse property is a highly desirable yet notoriously difficult feature in information theory, as it establishes a sharp boundary for achievable performance~\cite{cover1999elements}. Notably, the strong converse for channel discrimination in the best-case setting remains a major open problem~\cite{fang2025towards}, even after recent breakthroughs on the generalized quantum Stein's lemma~\cite{hayashi2024generalized,lami2024solutiongeneralisedquantumsteins}. It is also not even known for adversarial state discrimination in the classical setting~\cite{brandao2020adversarial}. The main technical challenge arises from the regularization of the optimal exponent, so the strong converse has only been restricted for certain degenerate cases where additivity holds and regularization is unnecessary~\cite{Cooney2016,wilde2020amortized}. Remarkably, we resolve this problem for the adversarial channel discrimination in full generality: Eq.~\eqref{eq: Adversarial quantum Stein's lemma} holds valid independently of $\ve$. To the best of our knowledge, this is the first time the strong converse property has been proven for a regularized exponent in quantum information theory. This also serves as a prior to establish later application in the relative entropy accumulation theorem, which requires explicit finite-size estimation without dependence of $\varepsilon$.}

\smallskip
\paragraph{Chain rules and proof outline for Theorem~\ref{thm: adversarial quantum channel Steins lemma}.---}{In the best-case channel discrimination, adaptive and non-adaptive strategies were studied independently and shown to be characterized by the amortized channel divergence and the regularized channel divergence, respectively~\cite{WW2019}. The equivalence of these two settings was fully resolved by some of us in~\cite{fang2020chain}, where the chain rule property of quantum relative entropy (in the upper bound direction) was established. In this work, we directly demonstrate the collapse of adaptive and non-adaptive strategies in the adversarial setting, in a manner resembling the chain rule of~\cite{fang2020chain}, but this time in the lower bound direction. This result represents the first chain rule of its kind for quantum relative entropy and complements previous findings, together providing a complete quantum analogue of the classical chain rule for relative entropy~\cite[Theorem 2.5.3]{cover1999elements}.}

Before the formal statement, we introduce a few notations. The \emph{measured relative entropy} is defined as
$D_{\Meas} (\rho\|\sigma) := \sup_{(\cX,M)} D(P_{\rho,M}\|P_{\sigma,M})$~\cite{donald1986relative,hiai1991proper},
where $D$ is the Kullback-Leibler divergence and the supremum is over all finite sets $\cX$ and POVMs $M = \{M_x\}_{x \in \cX}$ satisfying $M_x \geq 0$ and $\sum_{x \in \cX} M_x = I$, and $P_{\rho,M}(x) := \tr[M_x \rho]$. The \emph{measured \Renyi divergence} is
$D_{\Meas, \alpha} (\rho\|\sigma) := \sup_{(\cX,M)} D_{\alpha}(P_{\rho,M}\|P_{\sigma,M})$~\cite{Berta2017},
where $D_{\alpha}$ is the classical \Renyi divergence. The \emph{sandwiched \Renyi divergence} is
$D_{\Sand,\alpha}(\rho\|\sigma) := \frac{1}{\alpha-1}\log\tr[\sigma^{\frac{1-\alpha}{2\alpha}}\rho\sigma^{\frac{1-\alpha}{2\alpha}}]^\alpha$~\cite{muller2013quantum,wilde2014strong},
if $\supp(\rho) \subseteq \supp(\sigma)$, and $+\infty$ otherwise.  
Define $D_{\Meas,\alpha}^{\inf}(\cN\|\cM)$ and $D_{\Sand,\alpha}^{\inf,\reg}(\cN\|\cM)$ analogously to the minimum output channel divergence, but using the measured and sandwiched \Renyi divergences, respectively. 
Denote $\PSD$ as the set of positive semidefinite operators. Then we have the
chain rules as follows.

\begin{lemma}(Chain rules.)\label{lem: chain rule}
Let $\cN \in \CPTP(A\!:\!B)$, $\cM \in \CP(A\!:\!B)$, $\rho \in \density(RA)$, and $\sigma \in \PSD(RA)$. Let $\rho^\cN_{RB}=\cN(\rho_{RA})$ and $\sigma^\cM_{RB} = \cM (\sigma_{RA})$ be the channel output states.  Then for any $\alpha \in (0,+\infty)$, 
\begin{align*}
D_{\Meas,\alpha}(\rho^\cN_{RB}\|\sigma^\cM_{RB}) \geq D_{\Meas,\alpha}(\rho_R\|\sigma_R) + D_{\Meas,\alpha}^{\inf}(\cN\|\cM).
\end{align*} 
Moreover, for any $\alpha \in [1/2,\infty)$,
\begin{align*}
D_{\Sand,\alpha}(\rho^\cN_{RB}\|\sigma^\cM_{RB}) \geq D_{\Sand,\alpha}(\rho_R\|\sigma_R) + D_{\Sand,\alpha}^{\inf,\reg}(\cN\|\cM).
\end{align*}
When $\alpha=1$, the chain rules represent the results for the measured and Umegaki relative entropies, respectively.
\end{lemma} 

These chain rules can be viewed as a strengthening of the data processing inequality under partial trace. Notably, they are also tight in the sense that 
$\inf_{\rho,\sigma \in \density(RA)} \left[D_{\Sand,\alpha}(\rho^\cN_{RB}\|\sigma^\cM_{RB}) - D_{\Sand,\alpha}(\rho_R\|\sigma_R)\right] = D_{\Sand,\alpha}^{\inf,\reg}(\cN\|\cM)$.
In other words, the amortized channel divergence coincides with the regularized divergence, serving as the analog of~\cite[Corollary 3]{fang2020chain} in the worst-case scenario. {However, the proof of these new chain rules requires entirely different techniques~\cite{fang2025sm}.}

{With the new chain rules in hand, we can now outline the proof of Theorem~\ref{thm: adversarial quantum channel Steins lemma}. The first step is to simplify the minimax in $\beta_{n,\ve}(\cN\|\cM)$ by relating it to the quantum hypothesis testing relative entropy $D_{\Hypo,\ve}(\rho\|\sigma):=-\log\{\tr[\sigma M]: \tr[\rho(I-M)] \leq \ve,\, 0 \leq M \leq I\}$, which is mathematically more tractable. To do this, we show that the sets $\sA_n$ and $\sB_n$ are convex via explicit construction of adaptive strategies~\cite{fang2025sm}, and apply~\cite[Lemma 31]{fang2024generalized}. This yields
$-\log \beta_{n,\ve}(\cN\|\cM) = D_{\Hypo,\ve}(\sA_n\|\sB_n):= \inf_{\rho_n \in \sA_n,\, \sigma_n \in \sB_n} D_{\Hypo,\ve}(\rho_n\|\sigma_n)$. Then the proof of Theorem~\ref{thm: adversarial quantum channel Steins lemma} consists of achievable and converse parts.}

{For the achievable part, we consider the non-adaptive strategies, where the adversary ignores the environmental systems $E_i$ and performs no updates between rounds. This means that the sets $\sA_n$ and $\sB_n$ reduce to
$\sA_n' := \{\cN^{\ox n}(\rho_n): \rho_n \in \density(A^{\ox n})\}$ and $\sB_n' := \{\cM^{\ox n}(\sigma_n): \sigma_n \in \density(A^{\ox n})\}$. Using semidefinite programming duality, we can show that the sets $\{\sA_n'\}_n$ and $\{\sB_n'\}_n$ meet all the structural assumptions to apply the recent generalized asymptotic equipartition property in~\cite[Theorem 25]{fang2024generalized}. This allows us to conclude that the optimal type-II error exponent achieves the rate $D^{\inf,\reg}(\cN\|\cM)$.}

{For the converse part, we show that no adaptive strategy can achieve a type-II error exponent smaller than the rate $D^{\inf,\reg}(\cN\|\cM)$. The key step is to recursively apply the chain rule to obtain a lower bound on $D_{\Hypo,\ve}(\sA_n\|\sB_n)$.}
Denote the joint states before the $n$-th use of the channel by $\rho'_n := \tr_{R_n}\circ\;\cP^{n}\circ  \prod_{i=1}^{n-1}  \cU\circ \cP^i$ and $\sigma'_n := \tr_{R_n}\circ\;\cQ^{n}\circ \prod_{i=1}^{n-1}  \cV\circ \cQ^i$. Denote $\rho_n := \rho[\{\cP^i\}_{i=1}^n]$ and $\sigma_n := \sigma[\{\cQ^i\}_{i=1}^n]$. 
Then we have that $\rho_n = \tr_{R_nE_n} \circ \,\cU (\rho'_n) = \cN(\rho_n')$ and $\sigma_n = \tr_{R_nE_n} \circ \,\cV(\sigma_n') = \cM(\sigma_n')$. 
Note that for any $\cP^n$ and $\cQ^n$ we always have $\tr_{A_nR_n} \circ \, \cP^{n} = \tr_{E_{n-1}R_{n-1}}$ and $\tr_{A_nR_n} \circ \, \cQ^{n} = \tr_{E_{n-1}R_{n-1}}$.
This gives the relations that $\tr_{A_n}(\rho_n') = \rho_{n-1}$ and $\tr_{A_n}(\sigma_n') = \sigma_{n-1}$.
Applying the chain rule in Lemma~\ref{lem: chain rule} to $\rho_n$ and $\sigma_n$, we have $D_{\Sand,\alpha}(\rho_n\|\sigma_n)\geq D_{\Sand,\alpha}(\rho_{n-1}\|\sigma_{n-1}) + D^{\inf,\reg}_{\Sand,\alpha}(\cN\|\cM)$.
Recursively applying this relation $n$ times, we get $D_{\Sand,\alpha}(\rho_n\|\sigma_n) \geq n D^{\inf,\reg}_{\Sand,\alpha}(\cN\|\cM)$.
As this holds for any operations $\cP^i$ and $\cQ^i$, we get $D_{\Sand,\alpha}(\sA_n\|\sB_n) \geq n D^{\inf,\reg}_{\Sand,\alpha}(\cN\|\cM)$.
Using the relation of hypothesis testing relative entropy and the Petz \Renyi divergence in~\cite[Proposition 3]{qi2018applications} and the fact that Petz \Renyi divergence is no smaller than the sandwiched \Renyi divergence~\cite[Eq.(4.88)]{Tomamichel2015b}, we have for any $\alpha \in [1/2,1)$ and $\ve \in (0,1)$, that
$D_{\Hypo,\ve}(\sA_n\|\sB_n) \geq D_{\Sand,\alpha}(\sA_n\|\sB_n) + f(\ve,\alpha)$
where $f(\ve,\alpha) = -\log (\ve)\alpha/(\alpha-1)$.
Putting things together, we have the relation that
$D_{\Hypo,\ve}(\sA_n\|\sB_n) \geq n D^{\inf,\reg}_{\Sand,\alpha}(\cN\|\cM) + f(\ve,\alpha)$.
Taking the limits on both sides, we get
$\liminf_{n\to \infty} \frac{1}{n} D_{\Hypo,\ve}(\sA_n\|\sB_n) \geq \sup_{\alpha \in [1/2,1)}D^{\inf,\reg}_{\Sand,\alpha}(\cN\|\cM)$.
Finally, we relate the right-hand side to $D^{\inf,\reg}(\cN\|\cM)$ by noting that
$\sup_{\alpha \in [1/2,1)} D^{\inf,\reg}_{\Sand,\alpha}(\cN\|\cM) \geq \sup_{\alpha \in [1/2,1)} D^{\inf,\reg}_{\Meas,\alpha}(\cN\|\cM) = D^{\inf,\reg}(\cN\|\cM)$, where the inequality follows as $D_{\Meas,\alpha}(\rho\|\sigma) \leq D_{\Sand,\alpha}(\rho\|\sigma)$ for $\alpha \in [1/2,1)$, and the equality is a consequence of~\cite[Lemma 27, 28]{fang2024generalized}, applied to the image sets of the channels $\sA_n'$ and  $\sB_n'$. This concludes the proof of Theorem~\ref{thm: adversarial quantum channel Steins lemma}.

\smallskip

\paragraph{Relative entropy accumulation.---} 
The entropy accumulation theorem~\cite{dupuis2020entropy,metger2022generalised} is a technique to find bounds on the operationally relevant uncertainty present in the outputs of a sequential process as a sum of the worst case uncertainties of each step. It has been widely used in quantum cryptography~\cite{arnon2018practical,liu2018device,george2022finite}.
More specifically, a variant for the max-entropy $H_{\max}^{\ve}$ states that~\cite{metger2022generalised} for channels $\cN_{i} \in \CPTP(Y_{i-1}\!\!:\!\! Y_{i} S_i C_i)$, we have
$H_{\max}^{\ve}(S_1 \dots S_n | C_1 \dots C_n)_{\cN_n \circ \cdots \circ \cN_1(\rho_{Y_0})} \leq \sum_{i=1}^n \sup_{\omega_{Y_{i-1}}} H(S_i |C_i)_{\cN_i(\omega)} + O(\sqrt{n})$,
where $H_{\max}^{\ve}(S|C)_{\rho} \approx -\inf_{\sigma \in \density(C)} D_{\Hypo, \ve}(\rho_{SC} \| I_{S} \otimes \sigma_C)$\footnote{The exact definition is not in terms of $D_{\Hypo, \ve}$ but a closely related smoothed version of $D_{\Sand, 1/2}$. See Supplemental Materials for more details.
} and $H(S|C)_{\rho} = -\inf_{\sigma \in \density(C)} D(\rho_{SC} \| I_{S} \otimes \sigma_C)$. 
This naturally raises the question of whether such a statement can be generalized to divergences between arbitrary sequential processes of channels, rather than being restricted only to entropies (which corresponds to choosing $\cM$ in Figure~\ref{fig: adversarial setting} to be a replacer channel). This was first asked as an open question in~\cite{metger2022generalised} for the max-relative entropy. By extending the adversarial channel discrimination framework in Figure~\ref{fig: adversarial setting} to allow different channels $\cN_i$ and $\cM_i$ at each round, we can unify the relative entropy accumulation within this broader framework. This new perspective enables us to establish a relative entropy accumulation theorem for $D_{\Hypo, \ve}$ (a smoothed form of the min-relative entropy), giving an answer to the dual formulation of this open question. Specifically, we generalize the converse part of Theorem~\ref{thm: adversarial quantum channel Steins lemma} in two ways: we allow the channels applied at different steps to vary, and we compute explicit finite-size bounds. 

\begin{theorem}(Relative entropy accumulation theorem.)
Let $\cN_i \in \CPTP(A_i\!:\!A_{i+1}B_i), \cM_i \in \CP(A_i\!:\!A_{i+1}B_i)$ and $\rho, \sigma \in \density(A_1)$. Let $\ve \in (0,1)$, then it holds that
\begin{align*}
D&_{\Hypo, \ve} \bigg(\tr_{A_{n+1}} \circ \prod_{i=1}^n \cN_{i} (\rho_{A_1}) \bigg\| \tr_{A_{n+1}} \circ \prod_{i=1}^n \cM_{i}(\sigma_{A_1})\bigg)\\
& \geq \sum_{i=1}^n D^{\inf,\reg}(\tr_{A_{i+1}} \circ \cN_i \| \tr_{A_{i+1}} \circ \cM_i) - O(n^{2/3} \log n). 
\end{align*}
\end{theorem}

The proof makes important use of the chain rules in Lemma~\ref{lem: chain rule}. Moreover, choosing the channels $\cM_i$ to be replacer channels, we recover a slightly weaker version of the $H_{\max}^{\ve}$ entropy accumulation statement previously mentioned. We leave it as an open question for future work whether this new proof technique can lead to better entropy accumulation theorems.

\smallskip

\paragraph{Discussion.---} {We introduced the adversarial quantum channel discrimination and established its quantum Stein's lemma. This extreme framework complements and completes existing studies on the best-case setting, sets the boundary for the ultimate limits of channel discrimination, and establishes a new connnection between quantum information theory and quantum cryptography. Notably, the optimal error exponent can be efficiently computed despite its regularization, and the strong converse property holds in general. These features are unique and unprecedented in quantum information theory. Technically, we introduced the minimum output channel divergence and established its chain rules. This tool has recently been used in~\cite{arqand2025marginal} for the security analysis of quantum cryptographic protocols. Given the fundamental importance of quantum channel discrimination and our new connection to quantum cryptography, we anticipate that the new framework and tools developed will open new directions for future investigations.}

\smallskip

\begin{acknowledgements}
\paragraph{ Acknowledgments.} K.F. is supported by the National Natural Science Foundation of China (Grant No. 92470113 and 12404569), the Shenzhen Science and Technology Program (Grant No. JCYJ20240813113519025), the Shenzhen Fundamental Research Program (Grant No. JCYJ20241202124023031), the General R\&D Projects of 1+1+1 CUHK-CUHK(SZ)-GDST Joint Collaboration Fund (Grant No. GRDP2025-022), and the University Development Fund (Grant No. UDF01003565). O.F. acknowledges support by the European Research Council (ERC Grant AlgoQIP, Agreement No. 851716), by the European Union’s Horizon research and innovation programme under the project VERIqTAS (Grant Agreement No 101017733) and the project ``Quantum Security Networks Partnership" QSNP (Grant Agreement No 101114043), and by the Agence Nationale de la Recherche under the Plan France 2030 with the reference ANR-22-PETQ-0009.
\end{acknowledgements}

\nocite{umegaki1962conditional,muller2013quantum,wilde2014strong,datta2009min,renner2005security,petz1986quasi,donald1986relative,hiai1991proper,Berta2017,fang2024generalized,watrous2018theory,WW2019,fang2020chain,cover1999elements,arqand2025marginal,wilde2020amortized,leditzky2018approaches,qi2018applications,Tomamichel2015b,fang2025efficient,dupuis2020entropy,metger2022generalised,tomamichel2014relating,dupuis2019entropySecondOrder,carlen2010trace,dupuis2014generalized,watrous2009semidefinite,arnon2018practical,liu2018device,george2022finite}

\bibliographystyle{apsrev4-1}
\bibliography{Bib}

\clearpage

\onecolumngrid
\begin{center}
\vspace*{.5\baselineskip}
{\textbf{\large Supplemental Materials}}\\[1pt] \quad \\
\end{center}

\vspace{2cm}
In this Supplemental Material, we provide more detailed expositions, proofs and discussions of the results in the main text. We may reiterate some of the steps to ensure that the Supplemental Material are explicit and self-contained.

{
\hypersetup{linkcolor=black}
\tableofcontents
}

\section{Preliminaries}

\subsection{Notations and quantum divergences}

Let $\density(A)$ denote the set of all density operators on a finite-dimensional Hilbert space $\cH_A$. Let $\mathscr{L}(A)$ represent the set of all linear operators. Let $\HERM(A)$, $\PSD(A)$ and $\PD(A)$ be the set of all hermitian operators, positive semidefinite operators and positive definite operators on $\cH_A$, respectively. The set of completely positive maps from $\mathscr{L}(A)$ to $\mathscr{L}(B)$ is denoted by $\CP(A\!:\!B)$. A quantum channel $\cN_{A\to B}$ is a linear map from $\mathscr{L}(A)$ to $\mathscr{L}(B)$ that is both completely positive and trace-preserving. The set of all such maps is denoted by $\CPTP(A\!:\!B)$.  A quantum divergence is a functional $\DD:\density \times \PSD \to \RR$ that satisfies the data-processing inequality, characterizing the ``distinguishability'' or ``distance'' between two quantum states. There are a few quantum divergences used throughout this work.

\begin{definition}(Umegaki relative entropy~\cite{umegaki1962conditional}.)
For any $\rho\in \density$ and $\sigma \in \PSD$, the Umegaki relative entropy is defined by
\begin{align}\label{eq: Umegaki}
    D(\rho\|\sigma):= \tr[\rho(\log \rho - \log \sigma)],
\end{align}
if $\supp(\rho) \subseteq \supp(\sigma)$ and $+\infty$ otherwise.
\end{definition}

\begin{definition}(Sandwiched \Renyi divergence~\cite{muller2013quantum,wilde2014strong}.)
Let $\alpha \in (0,1) \cup (1,+\infty)$. For any $\rho\in \density$ and $\sigma \in \PSD$, the sandwiched \Renyi divergence is defined by
\begin{align}\label{eq: Sandwiched}
    D_{\Sand,\alpha}(\rho\|\sigma) := \frac{1}{\alpha-1}\log\tr\left[\sigma^{\frac{1-\alpha}{2\alpha}}\rho\sigma^{\frac{1-\alpha}{2\alpha}}\right]^\alpha,
\end{align}
if $\supp(\rho) \subseteq \supp(\sigma)$, and $+\infty$ otherwise.
\end{definition}

When $\alpha \to 1$, $D_{\Sand,\alpha}$ converge to the Umegaki relative entropy~\cite{muller2013quantum,wilde2014strong},
\begin{align}\label{eq: state Renyi continuous}
\lim_{\alpha \to 1} D_{\Sand,\alpha}(\rho\|\sigma) = D(\rho\|\sigma).
\end{align}
When $\alpha \to \infty$, the sandwiched \Renyi divergence converges to the max-relative entropy~\cite{datta2009min,renner2005security},
\begin{align}\label{eq: definition of Dmax}
\lim_{\alpha \to \infty} D_{\Sand,\alpha}(\rho\|\sigma) = D_{\max}(\rho\|\sigma):= \log\inf\big\{t \in \RR \;:\; \rho \leq t\sigma \big\}\;,
\end{align}
if $\supp(\rho) \subseteq \supp(\sigma)$ and $+\infty$ otherwise. 

\begin{definition}(Petz \Renyi divergence~\cite{petz1986quasi}.)
Let $\alpha \in (0,1) \cup (1,+\infty)$. For any $\rho\in \density$ and $\sigma \in \PSD$, the Petz \Renyi divergence is defined by
\begin{align}\label{eq: Petz}
    D_{\Petz,\alpha}(\rho\|\sigma) := \frac{1}{\alpha-1}\log\tr\left[\rho^\alpha\sigma^{1-\alpha}\right],
\end{align}
if  $\supp(\rho) \subseteq \supp(\sigma)$, and $+\infty$ otherwise. 
\end{definition}

\begin{definition}(Hypothesis testing relative entropy.)
Let $\ve \in [0,1]$. For any $\rho\in \density$ and $\sigma \in \PSD$, the quantum hypothesis testing relative entropy is defined by $D_{\Hypo, \ve}(\rho\|\sigma) := -\log \beta_{\ve}(\rho\|\sigma)$
where
\begin{align}
    \beta_\ve(\rho\|\sigma): = \min_{0\leq M \leq I} \left\{\tr[\sigma M]: \tr[\rho(I-M)] \leq \ve\right\}.
\end{align}
\end{definition}

\begin{definition}(Measured relative entropy~\cite{donald1986relative,hiai1991proper}.)
For any $\rho \in \density$, $\sigma \in \PSD$, the measured relative entropy is defined by
\begin{align}
D_{\Meas} (\rho\|\sigma) := \sup_{(\cX,M)} D(P_{\rho,M}\|P_{\sigma,M}),
\end{align}
where $D$ is the Kullback–Leibler divergence and the optimization is over finite sets $\cX$ and positive operator valued measures $M$ on $\cX$ such that $M_x \geq 0$ and $\sum_{x \in \cX} M_x = I$, $P_{\rho,M}$ is a measure on $\cX$ defined via the relation $P_{\rho,M}(x) = \tr[M_x\rho]$ for any $x \in \cX$.    
\end{definition}

\begin{definition}(Measured \Renyi divergence~\cite{Berta2017}.)
Let $\alpha \in (0,1) \cup (1,\infty)$.  For any $\rho \in \density$ and $\sigma \in \PSD$, the {measured \Renyi divergence} is defined as
\begin{align}\label{eq: definition DM alpha}
D_{\Meas, \alpha} (\rho\|\sigma) := \sup_{(\cX,M)} D_{\alpha}(P_{\rho,M}\|P_{\sigma,M}),
\end{align}
where $D_{\alpha}$ is the classical \Renyi divergence. 
\end{definition}
When $\alpha \to 1$, the measured \Renyi divergence converges to the measured relative entropy. 

\subsection{Minimum output channel divergence}

The quantum divergence between two quantum states can be naturally extended to two sets of quantum states. 

\begin{definition}(Quantum divergence between two sets of states.)
For any subsets $\sA \subseteq \density$ and $\sB \subseteq \PSD$, their divergence is defined as 
\begin{align}
\DD(\sA\|\sB) := \inf_{\rho \in \sA, \sigma \in \sB} \DD(\rho\|\sigma),
\end{align} 
which represents the minimum ``distance'' between the sets.
\end{definition}

The quantum divergence can also be extended to quantum channels.

\begin{definition}(Minimum output channel divergence.)
    Let $\cN \in \CPTP(A\!:\!B)$ and $\cM \in \CP(A\!:\!B)$. The {minimum output channel divergence} is defined by
\begin{align}
\DD^{\inf}(\cN\|\cM) := \inf_{\rho, \sigma \in \density(A)} \DD(\cN(\rho)\| \cM(\sigma)),
\end{align}
which conceptually captures the worst-case scenario where different test states are chosen to minimize the distinguishability between the given channels.
\end{definition}
 It is also useful to see this as a divergence between two sets of quantum states, expressed as 
\begin{align}
    \DD^{\inf}(\cN\|\cM) = \DD(\cN(\density)\|\cM(\density)),
\end{align}
where $\cL(\density) := \{\cL(\rho): \rho \in \density\}$ denotes the image set of $\density$ under the linear map $\cL$. We also define the regularized channel divergence as
\begin{align}
\DD^{\inf,\reg}(\cN\|\cM) := \lim_{n \to \infty} \frac{1}{n} \DD^{\inf}(\cN^{\ox n}\|\cM^{\ox n}),
\end{align}
which accounts for the asymptotic behavior of the channel divergence over multiple uses.

\begin{remark}
    As proved in Section~\ref{sec: Proof of the adversarial quantum Stein's lemma} that the image sets satisfy all assumptions in~\cite[Assumption 24]{fang2024generalized}.
    Note that the measured \Renyi divergence coincides with the sandwiched \Renyi divergence at $\alpha \in \{1/2, +\infty\}$. Thus, by applying the superadditivity in~\cite[Lemma 21]{fang2024generalized} and the subadditivity in~\cite[Lemma 26]{fang2024generalized}, we conclude that the minimum output channel divergence is additive for these cases, that is,
    \begin{align}
        D^{\inf}_{\Sand,1/2}(\cN_1\ox \cN_2\|\cM_1\ox \cM_2) & = D^{\inf}_{\Sand,1/2}(\cN_1\|\cM_1) + D^{\inf}_{\Sand,1/2}(\cN_2\|\cM_2)\\
        D_{\max}^{\inf}(\cN_1\ox \cN_2\|\cM_1\ox \cM_2) & = D_{\max}^{\inf}(\cN_1\|\cM_1) + D_{\max}^{\inf}(\cN_2\|\cM_2),
    \end{align}
    where the first equality recovers Watrous' result~\cite[Corollary 3.60]{watrous2018theory}. 
\end{remark}

\section{Chain rules}
\label{sec: Tightness of the chain rules}


{In the ``best-case'' channel discrimination, adaptive and non-adaptive strategies were studied independently and shown to be characterized by the amortized channel divergence and the regularized channel divergence, respectively~\cite{WW2019}. The equivalence of these two settings was fully resolved by some of us~\cite{fang2020chain}, where the chain rule property of quantum relative entropy (in the upper bound direction) was established. In this work, we directly demonstrate the collapse of adaptive and non-adaptive strategies in the adversarial setting, in a manner resembling the chain rule of~\cite{fang2020chain}, but this time in the lower bound direction. This is the first chain rule for quantum relative entropy in this direction and complements the previous result, together providing a complete analogue of the classical chain rule of relative entropy~\cite[Theorem 2.5.3]{cover1999elements}. Proving this new chain rule in the lower bound direction requires completely different techniques: specifically, we establish it using properties of measured relative entropy, whereas the chain rule in~\cite{fang2020chain} relied on smoothed max-relative entropy.}

\subsection{Proof of the chain rules}

The proof of these chain rules requires the notion of (reverse) polar sets. 

\begin{definition}
Let $\cvxset \subseteq \HERM$ be a convex set. Its polar set is defined by
\begin{align}
\cvxset^\circ:= \{X: \tr[XY] \leq 1,  \forall\, Y\in \cvxset\} = \{X : h_{\cvxset}(X) \leq 1\}
\end{align} 
where $h_{\cvxset}$ is the support function of $\cvxset$:
\begin{align}
    h_{\cvxset}(\omega) := \sup_{\sigma \in \cvxset} \tr[\omega \sigma].
\end{align}
Let $\polarPSD{\cvxset}:= \cvxset^\circ \cap \PSD$ and  $\polarPD{\cvxset}:= \cvxset^\circ \cap \PD$ be the intersections with positive semidefinite operators and positive definite operators, respectively.

Similarly, the reverse polar set of $\cvxset$ is defined as:
\begin{align}
    \cvxset^{\star} &:= \{X : \tr[XY] \geq 1, \forall\, Y \in \cvxset \} = \{X : \revh_{\cvxset}(X) \geq 1\}
\end{align}
where $\revh_{\cvxset}$ is the reverse support function of $\cvxset$:
\begin{align}
    \revh_{\cvxset}(\omega) := \inf_{\sigma \in \cvxset} \tr[\omega \sigma].
\end{align}
\end{definition}

\begin{lemma}(Chain rules.)\label{lem: chain rule DM}
Let $\rho \in \density(RA)$, $\sigma \in \PSD(RA)$, $\cN \in \CPTP(A\!:\!B)$ and $\cM \in \CP(A\!:\!B)$. Let $\rho^\cN_{RB}=\cN(\rho_{RA})$ and $\sigma^\cM_{RB} = \cM (\sigma_{RA})$ be the channel output states.  Then it holds, for any $\alpha \in (0,+\infty)$, that
\begin{align}\label{eq: chain rule DM}
D_{\Meas,\alpha}(\rho^\cN_{RB}\|\sigma^\cM_{RB}) \geq D_{\Meas,\alpha}(\rho_R\|\sigma_R) + D_{\Meas,\alpha}^{\inf}(\cN\|\cM).
\end{align} 
Moreover, it holds, for any $\alpha \in [1/2,\infty)$, that
\begin{align}\label{eq: chain rule DS}
D_{\Sand,\alpha}(\rho^\cN_{RB}\|\sigma^\cM_{RB}) \geq D_{\Sand,\alpha}(\rho_R\|\sigma_R) + D_{\Sand,\alpha}^{\inf,\reg}(\cN\|\cM).
\end{align}
When $\alpha=1$, the chain rules represent the results for the measured and Umegaki relative entropies, respectively.
\end{lemma}

\begin{proof}
The proof utilizes the superadditivity of the divergence between two sets of quantum states, as established in~\cite[Lemma 21, 23]{fang2024generalized}. To this end, we consider the following sets:
\begin{alignat}{3}
 \sA_1 & = \{\rho_R\}, \quad \sA_2 &= \cN(\density), \quad \sA_3 & = \{\rho_{RB}^{\cN}\},\\
 \sB_1 & = \{\sigma_R\}, \quad \sB_2 & = \cM(\density), \quad \sB_3 & = \{\sigma_{RB}^\cM\}.
\end{alignat}
{We apply semidefinite programming duality to} verify that they meet the required assumptions. We do this for $\{\sA_1, \sA_2, \sA_3\}$ and the same argument works for $\{\sB_1, \sB_2, \sB_3\}$. 
For any $Y_B \in \polarPSD{(\sA_{2})}$, we have $\tr[Y_B \cN(\rho)] \leq 1$ for any $\rho \in \density(A)$. This implies that $\cN^\dagger(Y_B) \leq I_A$, with $\cN^\dagger$ being the adjoint map of $\cN$. Therefore, for any $X_R \in \polarPSD{(\sA_{1})}$ and $Y_B \in \polarPSD{(\sA_{2})}$, we have the following equations,
\begin{align}
    \tr[(X_R \ox Y_B) \cN_{A\to B}(\rho_{RA})] & = \tr[(X_R \ox \cN^\dagger(Y_B)) \rho_{RA}]\\
    & \leq \tr[(X_R \ox I_A) (\rho_{RA})]\\
    & = \tr[X_R \rho_R]\\
    & \leq 1.
\end{align}
This implies that $X_R \ox Y_B \in \polarPSD{(\sA_3)}$ and therefore $\polarPSD{(\sA_{1})} \ox \polarPSD{(\sA_{2})} \subseteq \polarPSD{(\sA_3)}$. Applying the superadditivity in ~\cite[Lemma 21, 23]{fang2024generalized}, we have the asserted result in Eq.~\eqref{eq: chain rule DM} for $\alpha \in (0,1]$. The proof for $\alpha \in (1,+\infty)$ follows in the same way. For $\alpha \in (1,\infty)$, we work with the reverse polar sets instead: if $Y_B \in (\sA_2)^{\star}_{\plus}$, then $\cN^{\dagger}(Y_B) \geq I_{A}$ and as a result, if $X_R \in (\sA_1)^{\star}_{\plus}$, we have $X_R \otimes Y_B \in (\sA_3)^{\star}_{\plus}$ and we can similarly apply~\cite[Lemma 21]{fang2024generalized}. The result in Eq.~\eqref {eq: chain rule DS} is a direct consequence of Eq.~\eqref{eq: chain rule DM}. More specifically, we have that for any $\alpha \in [1/2,\infty)$,
\begin{align}
    D_{\Sand,\alpha}(\rho^\cN_{RB}\|\sigma^\cM_{RB}) & = \lim_{n\to \infty} \frac{1}{n} D_{\Meas,\alpha}((\rho^\cN_{RB})^{\ox n}\|(\sigma^\cM_{RB})^{\ox n})\\
    & \geq \lim_{n\to \infty} \frac{1}{n} D_{\Meas,\alpha}((\rho_R)^{\ox n}\|(\sigma_{B})^{\ox n}) + \lim_{n\to \infty} \frac{1}{n} D_{\Meas,\alpha}^{\inf}(\cN^{\ox n}\|\cM^{\ox n})\\
    & = D_{\Sand,\alpha}(\rho_R\|\sigma_R) + D_{\Sand,\alpha}^{\inf,\reg}(\cN\|\cM),
\end{align}
where the first line follows from~\cite[Lemma 16, 17]{fang2024generalized}, the second line follows from Eq.~\eqref{eq: chain rule DM} and the last line follows from~\cite[Lemma 28]{fang2024generalized}.
\end{proof}

\bigskip
{These chain rules can be viewed as a strengthening of the data processing inequality under partial trace. It can also be seen as an operational reformulation of the superadditivity of the measured relative entropy from~\cite[Lemma 21]{fang2024generalized}. The importance of this operational perspective is evidenced by the recent application of our chain rule in the security analysis of quantum cryptography by an independent group~\cite{arqand2025marginal}, and we anticipate that the chain rule properties developed here will enable further applications in the future.}

\subsection{Tightness of the chain rules}

In the following, we introduce the notion of the amortized minimum output channel divergence and show that it coincides with the regularized divergence,  being an analog result for the best-case channel divergence~\cite{fang2020chain}. This, in turn, demonstrates the tightness of our chain rule properties.

Similar to the amortized channel divergence used in the existing literature~\cite{wilde2020amortized}, we can define the minimum output version as follows.
Let $\DD$ be a quantum divergence between states. Let $\cN \in \CPTP(A\!:\!B)$ and $\cM \in \CP(A\!:\!B)$. Then the amortized minimum output channel divergence is defined by
\begin{align}
    \DD^{\inf, \amo}(\cN\|\cM):= \inf_{\substack{\rho\in \density(RA)\\ \sigma \in \density(RA)}} \DD(\cN_{A\to B}(\rho_{RA})\|\cM_{A\to B}(\sigma_{RA})) - \DD(\rho_R\|\sigma_R).
\end{align}

\begin{lemma}
Let $\alpha \in [1/2,\infty)$. For any $\cN\in \CPTP(A\!:\!B)$, $\cM \in \CP(A\!:\!B)$, it holds that
\begin{align}
D_{\Sand,\alpha}^{\inf, \amo}(\cN\|\cM) & = D_{\Sand,\alpha}^{\inf,\reg}(\cN\|\cM).
\end{align}     
Equivalently, for any $\cN\in \CPTP(A\!:\!B)$, $\cM \in \CP(A\!:\!B)$ and any $\ve \in (0,1)$, there exists $\rho,\sigma\in \density(RA)$, such that \begin{align}
    D_{\Sand,\alpha}^{\inf,\infty}(\cN\|\cM) \leq D_{\Sand,\alpha}(\cN_{A\to B}(\rho_{RA})\|\cM_{A\to B}(\sigma_{RA})) - D_{\Sand,\alpha}(\rho_R\|\sigma_R)  \leq D_{\Sand,\alpha}^{\inf,\infty}(\cN\|\cM) + \ve.
\end{align}
\end{lemma}

\begin{proof}
We prove the result for the quantum relative entropy ($\alpha = 1$) and the same argument works for the sandwiched \Renyi divergence of order $\alpha \neq 1$ as well. Note that the chain rule property in Lemma~\ref{lem: chain rule DM} is equivalent to $D^{\inf, \amo}(\cN\|\cM) \allowbreak \geq D^{\inf,\reg}(\cN\|\cM)$.
Now we prove the reverse direction. For this, we will first show the superadditivity of the amortized divergence under tensor product. That is, 
\begin{align}\label{eq: amortized worst case superadditivity}
D^{\inf, \amo}(\cN_1 \ox \cN_2 \|\cM_1\ox \cM_2) \geq D^{\inf, \amo}(\cN_1\|\cM_1) + D^{\inf, \amo}(\cN_2\|\cM_2),
\end{align}
for any quantum channels $\cN_1 \in \CPTP(A_1\!:\!B_1)$, $\cN_2\in \CPTP(A_2\!:\!B_2)$, CP maps $\cM_1 \in \CP(A_1\!:\!B_1)$, $\cM_2\in \CP(A_2\!:\!B_2)$. 
To see this, let $(\rho_{RA_1A_2}, \sigma_{RA_1A_2})$ be any feasible solution to the optimization on the left-hand side of Eq.~\eqref{eq: amortized worst case superadditivity} and denote its corresponding objective value by
\begin{align}
    \delta_{12} := D(\cN_1 \ox \cN_2(\rho_{RA_1A_2})\|\cM_1 \ox \cM_2(\sigma_{RA_1A_2})) - D(\rho_R\|\sigma_R).
\end{align}
Let $\omega_{RA_2B_1} = \cN_1(\rho_{RA_1A_2})$ and $\gamma_{RA_2B_1} = \cM_1(\sigma_{RA_1A_2})$. We can check that $(\rho_{RA_1},\sigma_{RA_1})$ and $(\omega_{RA_2B_1},\allowbreak \gamma_{RA_2B_1})$ are feasible solutions to the amortized divergence on the right-hand side of Eq.~\eqref{eq: amortized worst case superadditivity}, respectively, with the corresponding objective values by
\begin{align}
    \delta_1 := &  D(\cN_1(\rho_{RA_1})\|\cM_1(\sigma_{RA_1})) - D(\rho_R\|\sigma_R) \geq D^{\inf, \amo}(\cN_1\|\cM_1)\\
    \delta_2 := & D(\cN_2(\omega_{RA_2B_1})\|\cM_2(\gamma_{RA_2B_1})) - D(\omega_{RB_1}\|\gamma_{RB_1}) \geq D^{\inf, \amo}(\cN_2\|\cM_2).
\end{align}
Noting that $\omega_{RB_1} = \cN_1(\rho_{RA_1})$ and $\gamma_{RB_1} = \cM_1(\sigma_{RA_1})$, we have $\delta_{12} = \delta_1 + \delta_2$. This implies
\begin{align}
    \delta_{12} \geq D^{\inf, \amo}(\cN_1\|\cM_1) + D^{\inf, \amo}(\cN_2\|\cM_2).
\end{align}
As this holds for any feasible solution $(\rho_{RA_1A_2}, \sigma_{RA_1A_2})$, we have the asserted result in Eq.~\eqref{eq: amortized worst case superadditivity}.

By trivializing the reference system in the amortized divergence, we get $D^{\inf, \amo}(\cN\|\cM) \leq D^{\inf}(\cN\|\cM)$. Then we have
\begin{align}
    D^{\inf, \amo}(\cN\|\cM) \leq \frac{1}{n} D^{\inf, \amo}(\cN^{\ox n}\|\cM^{\ox n}) \leq \frac{1}{n} D^{\inf}(\cN^{\ox n}\|\cM^{\ox n}),
\end{align}
where the first inequality follows from Eq.~\eqref{eq: amortized worst case superadditivity}. As the above holds for any $n$, we can take $n\to \infty$ on the right-hand side and conclude that $D^{\inf, \amo}(\cN\|\cM) \leq D^{\inf,\reg}(\cN\|\cM)$. This completes the proof.
\end{proof}

\subsection{Counter-example to a potential improvement of the chain rule}
\label{sec: Counter-example to a potential improvement of the chain rule}

The quantum channel divergence studied in most existing literatures use the same test states for both channels, e.g., $\sup_{\rho \in \density(RA)} D(\cN_{A\to B}(\rho_{RA})\|\cM_{A\to B}(\rho_{RA}))$~\cite{leditzky2018approaches,WW2019}. So it may be expected that we can enhance the chain rules by using the same test states as well. However, we show here that this is not possible by giving a counter-example. That is, the chain rule cannot be enhanced to 
\begin{align}\label{eq: chain rule enhancement}
D_{\Meas}(\cN_{A\to B}(\rho_{RA})\|\cM_{A\to B}(\sigma_{RA})) \geq D_{\Meas}(\rho_R\|\sigma_R) + D_{\Meas}^{\inf'}(\cN\|\cM)
\end{align} 
where the channel divergence takes the same input state
\begin{align}
D_{\Meas}^{\inf'}(\cN\|\cM):= \inf_{\rho\in \density(A)} D_{\Meas}(\cN_{A\to B}(\rho_{A})\| \cM_{A \to B}(\rho_{A})),
\end{align}
{which is structurally closer to
the best-case divergence.}
To see this, consider the generalized amplitude damping (GAD) channel, which is defined as
\begin{align}
    \cA_{\gamma, N}(\rho)  = \sum_{i=1}^4 A_i \rho A_i^\dagger, \quad \gamma, N \in [0,1],
\end{align}
with the Kraus operator
\begin{alignat}{2}
    A_1 & = \sqrt{1-N}(\ket{0}\bra{0} + \sqrt{1-\gamma}\ket{1}\bra{1}), \qquad && A_2  = \sqrt{\gamma(1-N)} \ket{0}\bra{1},\\  
    A_3 & = \sqrt{N}(\sqrt{1-\gamma} \ket{0}\bra{0} + \ket{1}\bra{1}), \qquad
    && A_4 =  \sqrt{\gamma N} \ket{1}\bra{0}.
\end{alignat}

Using convex optimization, we can numerically evaluate each terms $D_{\Meas}(\cN_{A\to B}(\rho_{RA})\|\cM_{A\to B}(\sigma_{RA}))$, $D_{\Meas}(\rho_R\|\sigma_R)$ and $D_{\Meas}^{\inf'}(\cN\|\cM)$. Then in Figure~\ref{fig: chain rule enhancement}(a), we show that the channel divergence $D_{\Meas}^{\inf'}$ is subadditive under tensor product of channels. That is, it does not inherit the properties of the state divergence, making it not a suitable channel extension. Moreover, in Figure~\ref{fig: chain rule enhancement}(b), we show that the chain rule property in Eq.~\eqref{eq: chain rule enhancement} does not hold, as there are cases such that $y < x_2$ in the plot.

\begin{figure}[h]
    \centering
    \includegraphics[width=0.9\linewidth]{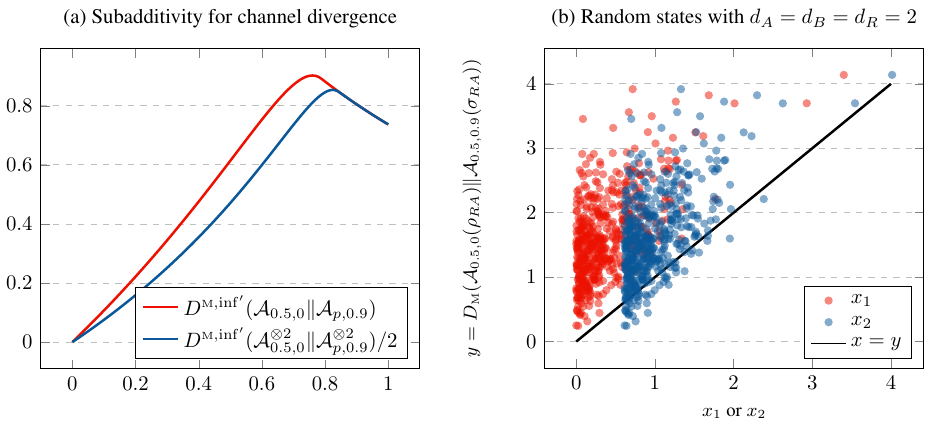}
    \caption{(a) Subadditivity for the channel divergence $D^{\Meas,\inf'}$ where $\cA_{0.5, 0}$ and $\cA_{p, 0.9}$ are the GAD channels and $p$ ranges from $0$ to $1$; (b) Random test for the chain rule property, where the quantum channels are chosen as $\cA_{0.5, 0}$ and $\cA_{0.5, 0.9}$, the quantum states are 500 randomly generated quantum states with real entries,  $x_1 = D_{\Meas}(\rho_R\|\sigma_R) + D_{\Meas}^{\inf}(\cA_{0.5, 0}\|\cA_{0.5, 0.9})$ and $x_2 = D_{\Meas}(\rho_R\|\sigma_R) + D_{\Meas}^{\inf'}(\cA_{0.5, 0}\|\cA_{0.5, 0.9})$.}
    \label{fig: chain rule enhancement}
\end{figure}

\section{Proof of the adversarial quantum Stein's lemma}
\label{sec: Proof of the adversarial quantum Stein's lemma}

Recall that after $n$ rounds of state generation, if the device behaves as $\cN$ and the adversary applies operations $\cP^i$, then the tester's overall state on $B_1\cdots B_n$ is given by:
\begin{align}
\rho[\{\cP^i\}_{i=1}^n] := \tr_{R_nE_n} \prod_{i=1}^{n} \left[\cU_{A_i \to B_i E_i} \circ \cP^i_{R_{i-1}E_{i-1} \to A_{i}R_{i}}\right].
\end{align}
Similarly, if the device is governed by $\cM$ and the internal operations by the adversary are given by $\cQ^i$, then the overall state is given by 
\begin{align}
\sigma[\{\cQ^i\}_{i=1}^n] := \tr_{R_nE_n} \prod_{i=1}^{n} \left[\cV_{A_i \to B_i E_i} \circ \cQ^i_{R_{i-1}E_{i-1} \to A_{i}R_{i}}\right].  
\end{align}
Due to limited knowledge of the device's internal workings and the adversary's strategies, the tester only has access to partial information, knowing that their possessed state belongs to one of two sets:
\begin{align}\label{eq: sets An}
    \sA_n := \{\rho[\{\cP^i\}_{i=1}^n] : \cP^i \in \CPTP(R_{i-1}E_{i-1}\!\!:\!\!A_iR_i), \forall R_i, \forall i\},
\end{align}
or
\begin{align}\label{eq: sets Bn}
    \sB_n := \{\sigma[\{\cQ^i\}_{i=1}^n] : \cQ^i \in \CPTP(R_{i-1}E_{i-1}\!\!:\!\!A_iR_i), \forall R_i, \forall i\},
\end{align}
where the adversary's internal memory $R_i$ may have arbitrarily large dimension.
In particular, if the adversary performs non-adaptive strategies, meaning they disregard the environmental systems $E_i$ and do not update their operations between rounds, then the sets $\sA_n$ and $\sB_n$ simplify to
\begin{align}
\sA_n' := \{\cN^{\otimes n}(\rho_n) : \rho_n \in \density(A^{\otimes n})\}, \quad
\sB_n' := \{\cM^{\otimes n}(\sigma_n) : \sigma_n \in \density(A^{\otimes n})\}.
\end{align}

Since non-adaptive strategies are a specific type of adaptive strategy, we have the inclusions that $\sA_n' \subseteq \sA_n \quad \text{and} \quad \sB_n' \subseteq \sB_n$.
This gives the relations for the type-I and type-II errors by $\alpha(\sA_n', M_n) \leq \alpha(\sA_n, M_n)$ and $\beta(\sB_n', M_n) \leq \beta(\sB_n, M_n)$.
So we have the general relation that 
\begin{align}
    \beta_{n,\ve}'(\cN\|\cM)  := \inf \{\beta(\sB_n', M_n): 0\leq M_n\leq I, \alpha(\sA_n', M_n) \leq \ve \} \leq  \beta_{n,\ve}(\cN\|\cM),
\end{align}
where the inequality holds because the right-hand side represents the infimum of a larger objective value over a smaller feasible set. 

\begin{theorem}(Adversarial quantum Stein's lemma.)\label{thm: adversarial quantum channel Steins lemma}
Let $\cN \in \CPTP(A\!:\!B)$ and $\cM \in \CP(A\!:\!B)$. 
Then for any $\ve \in (0,1)$, it holds that
\begin{align}\label{eq: Adversarial quantum Stein's lemma}
    \lim_{n\to \infty} - \frac{1}{n} \log \beta_{n, \ve}(\cN\|\cM) = D^{\inf,\reg}(\cN\|\cM).
\end{align}
Moreover, the Stein's exponent can be achieved by non-adaptive strategies, indicating that adaptive strategies by the adversary offer no advantage over non-adaptive ones.
\end{theorem}

\begin{proof}
{The first step is to simplify the minimax in $\beta_{n,\ve}(\cN\|\cM)$ by relating it to the quantum hypothesis testing relative entropy, which is mathematically more tractable. To do this, we show that the sets $\sA_n$ and $\sB_n$ are convex via explicit construction of adaptive strategies (see Lemma~\ref{lem: ad channel state convex} below), and apply~\cite[Lemma 31]{fang2024generalized}. This gives
\begin{align}
    -\log \beta_{n,\ve}(\cN\|\cM) & = D_{\Hypo,\ve}(\sA_n\|\sB_n).
\end{align}  
Similarly, it is clear that $\sA_n'$ and $\sB_n'$ are convex, and we also have
\begin{align}
        -\log \beta_{n,\ve}'(\cN\|\cM) & = D_{\Hypo,\ve}(\sA_n'\|\sB_n').
\end{align}
After this, the proof for the adversarial quantum Stein's lemma then contains two parts. The achievable part aims to show that 
\begin{align}
\limsup_{n\to \infty} \frac{1}{n} D_{\Hypo,\ve}(\sA_n'\|\sB_n') \leq D^{\inf,\reg}(\cN\|\cM),
\end{align}
by applying the generalized quantum asymptotic equipartition (AEP) property in~\cite[Theorem 25]{fang2024generalized} to $\sA_n'$ and $\sB_n'$. 
while the converse part makes use of the chain rule property to show that 
\begin{align}
\liminf_{n\to \infty} \frac{1}{n} D_{\Hypo,\ve}(\sA_n\|\sB_n) \geq D^{\inf,\reg}(\cN\|\cM).
\end{align}}

\emph{1) Achievabe part:} Using semidefinite programming duality, we can show that the sets $\{\sA_n'\}_n$ and $\{\sB_n'\}_n$ meet all the structural assumptions to apply the generalized AEP in~\cite[Theorem 25]{fang2024generalized}.
First, the set of all density matrices $\density$ is convex and compact, so $\sA_n'$ is also convex and compact. Since $\cN^{\ox n}$ and $\density$ are permutation invariant, we know that $\sA_n'$ is also permutation invariant. For any $\cN^{\ox m}(\rho_m) \in \sA_m'$ and $\cN^{\ox k}(\rho_k) \in \sA_k'$, we have 
\begin{align} 
    \cN^{\ox m}(\rho_m) \ox \cN^{\ox k}(\rho_k) = \cN^{\ox (m+k)}(\rho_m\ox \rho_k)\in \sA_{m+k}'.
\end{align} 
This implies $\sA_m' \ox \sA_k' \subseteq \sA_{m+k}'$. The support function of $\sA_n'$ is given by 
\begin{align} 
    h_{\sA'_n}(X_n) := \sup_{\rho_n \in \density} \tr \left[X_n \cN^{\ox n}(\rho_n) \right]= \sup_{\rho_n \in \density} \tr \left[(\cN^{\ox n})^\dagger (X_n) \rho_n \right] = \lambda_{\max}\left((\cN^{\ox n})^\dagger (X_n)\right),
\end{align}
where $\lambda_{\max}(\cdot)$ denotes the maximum eigenvalue.
Therefore, for any $X_m \in \PSD$ and $X_k\in \PSD$, we have
$h_{\sA'_{m+k}}(X_m \ox X_k)  = h_{\sA'_{m}}(X_m)h_{\sA'_{k}}(X_k)$,
by the multiplicativity of the maximum eigenvalue under tensor product. This proves that $\{\sA'_n\}_{n\in \NN}$ satisfy all the required assumptions, and the same argument works for $\{\sB_n'\}_{n\in \NN}$. This gives \begin{align} 
    \limsup_{n\to \infty} \frac{1}{n} D_{\Hypo,\ve}(\sA_n'\|\sB_n')  = D^{\inf,\reg}(\cN\|\cM),
\end{align} from~\cite{fang2024generalized} and proves the achievable part.

\emph{2) Converse part:} {Because the adversary's strategies can be adaptive, the sets $\sA_n$ and $\sB_n$ are not permutation invariant, so existing results such as the generalized AEP~\cite[Theorem 25]{fang2024generalized} do not directly apply. To address this, the new chain rules become crucial.} The first step is to lower bound $D_{\Hypo,\ve}(\sA_n\|\sB_n)$ by applying the chain rule recursively.
Denote the joint states before the $n$-th use of the channel by 
\begin{align} 
    \rho'_n := \tr_{R_n}\circ\;\cP^{n}\circ  \prod_{i=1}^{n-1}  \cU\circ \cP^i, \quad \text{and} \quad \sigma'_n := \tr_{R_n}\circ\;\cQ^{n}\circ \prod_{i=1}^{n-1}  \cV\circ \cQ^i.
\end{align} 
Denote $\rho_n := \rho[\{\cP^i\}_{i=1}^n]$ and $\sigma_n := \sigma[\{\cQ^i\}_{i=1}^n]$. 
Then we have that $\rho_n = \tr_{R_nE_n} \circ \,\cU (\rho'_n) = \cN(\rho_n')$ and $\sigma_n = \tr_{R_nE_n} \circ \,\cV(\sigma_n') = \cM(\sigma_n')$. 
Note that for any $\cP^n$ and $\cQ^n$ we always have $\tr_{A_nR_n} \circ \, \cP^{n} = \tr_{E_{n-1}R_{n-1}}$ and $\tr_{A_nR_n} \circ \, \cQ^{n} = \tr_{E_{n-1}R_{n-1}}$.
This gives the relations that $\tr_{A_n}(\rho_n') = \rho_{n-1}$ and $\tr_{A_n}(\sigma_n') = \sigma_{n-1}$.
Applying the chain rule in Lemma~\ref{lem: chain rule DM} to $\rho_n$ and $\sigma_n$, we have
\begin{align}
    D_{\Sand,\alpha}(\rho_n\|\sigma_n)\geq D_{\Sand,\alpha}(\rho_{n-1}\|\sigma_{n-1}) + D^{\inf,\reg}_{\Sand,\alpha}(\cN\|\cM).
\end{align}
Recursively applying this relation $n$ times, we get
\begin{align}
    D_{\Sand,\alpha}(\rho_n\|\sigma_n) \geq n D^{\inf,\reg}_{\Sand,\alpha}(\cN\|\cM).
\end{align}
As this holds for any operations $\cP^i$ and $\cQ^i$, we get
\begin{align}
    D_{\Sand,\alpha}(\sA_n\|\sB_n) \geq n D^{\inf,\reg}_{\Sand,\alpha}(\cN\|\cM).
\end{align}

Using the relation of hypothesis testing relative entropy and the Petz \Renyi divergence in~\cite[Proposition 3]{qi2018applications} and the fact that Petz \Renyi divergence is no smaller than the sandwiched \Renyi divergence~\cite[Eq.(4.88)]{Tomamichel2015b}, we have for any $\alpha \in [1/2,1)$ and $\ve \in (0,1)$, that
\begin{align}
D_{\Hypo,\ve}(\sA_n\|\sB_n) \geq D_{\Sand,\alpha}(\sA_n\|\sB_n) + (\log (1/\ve))\alpha/(\alpha-1).
\end{align}
Combining the above relations, we have 
\begin{align}
    D_{\Hypo,\ve}(\sA_n\|\sB_n) \geq n D^{\inf,\reg}_{\Sand,\alpha}(\cN\|\cM) + (\log (1/\ve))\alpha/(\alpha-1).
\end{align}
Taking the limits on both sides, we get
\begin{align}
\liminf_{n\to \infty} \frac{1}{n} D_{\Hypo,\ve}(\sA_n\|\sB_n) \geq \sup_{\alpha \in [1/2,1)}D^{\inf,\reg}_{\Sand,\alpha}(\cN\|\cM).
\end{align}
Finally, note that 
\begin{align} 
    \sup_{\alpha \in [1/2,1)} D^{\inf,\reg}_{\Sand,\alpha}(\cN\|\cM) \geq \sup_{\alpha \in [1/2,1)} D^{\inf,\reg}_{\Meas,\alpha}(\cN\|\cM) = D^{\inf,\reg}(\cN\|\cM),
\end{align} 
where the inequality follows as $D_{\Meas,\alpha}(\rho\|\sigma) \leq D_{\Sand,\alpha}(\rho\|\sigma)$ for $\alpha \in [1/2,1)$, and the equality is a consequence of~\cite[Lemma 27, 28]{fang2024generalized}, applied to the image sets of the channels $\sA_n'$ and $\sB_n'$.
\end{proof}

\begin{lemma}(Convexity.)\label{lem: ad channel state convex}
    The sets $\{\sA_n\}_{n\in \NN}$ and $\{\sB_n\}_{n\in \NN}$ defined in Eqs.~\eqref{eq: sets An} and~\eqref{eq: sets Bn} are convex.
\end{lemma}
\begin{proof}
We prove the assumptions for $\{\sA_n\}_{n\in \NN}$ and the same reasoning works for $\{\sB_n\}_{n\in \NN}$ as well.   
Let $\{\cP^i\}_{i=1}^n$ with systems $R_1, \dots, R_n$ and $\{\bar{\cP}^i\}_{i=1}^n$ with systems $R_1, \dots, R_n$ be two strategies and $\lambda \in [0,1]$. Note that we may assume both strategies have the same systems $R_i$ as we can always increase the dimension of the systems $R_i$ by extending the action of the channel in an arbitrary way without affecting the output.
Let us now define another strategy $\{\bar{\bar{\cP}}^i\}_{i=1}^n$ as follows. Let $\bar{\bar{R}}_i = R_i \otimes C$ for all $i=1,\dots,n$ where $C$ is a two-dimensional system. Then define 
\begin{align}
    \bar{\bar{\cP}}^1(\cdot) = \lambda \cP^1(\cdot) \ox \proj{0}_C + (1-\lambda) \bar{\cP}^1(\cdot) \ox \proj{1}_C
\end{align}
and for $i \geq 2$, define 
\begin{align}
    \bar{\bar{\cP}}^i(X) = (\cP^i \circ \cC_0(X)) \ox \proj{0}_C + (\bar{\cP}^i \circ \cC_1(X)) \ox \proj{1}_C
\end{align}
where $\cC_0(X) = \<0|_C X |0\>_C$ and $\cC_1(X) = \<1|_C X |1\>_C$. Since $\cC_0, \cC_1, \cP^i, \bar{\cP}^i$ are all CP maps, we know that $\bar{\bar{\cP}}^i$ is also a CP map. Moreover, if $\cP^i, \bar{\cP}^i$ are trace-preserving, then $\bar{\bar{\cP}}^i$ is also trace-preserving. It is easy to check the following relations,
\begin{alignat}{3}
    \cC_0 \circ \bar{\bar{\cP}}^1 & = \lambda \cP^1 \quad & \text{and} \quad \cC_1 \circ \bar{\bar{\cP}}^1 & = (1-\lambda) \bar{\cP}^1,\label{eq: ad convex tmp1}\\
    \cC_0 \circ \bar{\bar{\cP}}^i & = \cP^i \circ \cC_0 \quad & \text{and} \quad \cC_1 \circ \bar{\bar{\cP}}^i & = \bar{\cP}^i \circ \cC_1, \quad \forall i \geq 2.\label{eq: ad convex tmp2}
\end{alignat}
Noting that $\tr_C$ commutes with $\cU$ as they are acting on different systems, we have 
\begin{align}\label{eq: ad convex tmp3}
\tr_C \circ \prod_{i=1}^n \cU \circ \bar{\bar{\cP}}^i 
    & = \cU \circ \cP^n \circ \cC_0 \circ \prod_{i=1}^{n-1} \cU \circ \bar{\bar{\cP}}^i + \cU \circ \bar{\cP}^n \circ \cC_1 \circ \prod_{i=1}^{n-1} \cU \circ \bar{\bar{\cP}}^i.
\end{align}
Also noting that $\cC_0$ and $\cC_1$ both commute with $\cU$ as they are acting on different systems and using the relations in Eqs.~\eqref{eq: ad convex tmp1} and~\eqref{eq: ad convex tmp2}, we have 
\begin{align}
    \cU \circ \cP^n \circ \cC_0 \circ \prod_{i=1}^{n-1} \cU \circ \bar{\bar{\cP}}^i & = \lambda\, \cU \circ \cP^n \circ \prod_{i=1}^{n-1} \cU \circ {{\cP}}^i = \lambda\, \prod_{i=1}^{n} \cU \circ {{\cP}}^i,\\
    \cU \circ \bar{\cP}^n \circ \cC_1 \circ \prod_{i=1}^{n-1} \cU \circ \bar{\bar{\cP}}^i & = (1-\lambda)\, \cU \circ \cP^n \circ \prod_{i=1}^{n-1} \cU \circ {\bar{\cP}}^i = (1-\lambda)\, \prod_{i=1}^{n} \cU \circ {\bar{\cP}}^i.
\end{align}
Taking these into Eq.~\eqref{eq: ad convex tmp3}, we have 
\begin{align}
    \tr_{\bar{\bar{R}}_n E_n} \circ \prod_{i=1}^n \cU \circ \bar{\bar{\cP}}^i = \lambda\, \tr_{R_n E_n} \circ  \prod_{i=1}^{n} \cU \circ {{\cP}}^i + (1-\lambda)\, \tr_{R_n E_n} \circ  \prod_{i=1}^{n} \cU \circ {\bar{\cP}}^i
\end{align}
This shows that any mixture of the reduced states on $B_1\cdots B_n$ by the strategies $\{\cP^i\}_{i=1}^n$ and $\{\bar{\cP}^i\}_{i=1}^n$ is also given by the reduced state of another strategy $\{\bar{\bar{\cP}}^i\}_{i=1}^{n}$ which proves the convexity of $\sA_n$. 
\end{proof}

\section{Computational aspect of the Stein's exponent} 

{The minimum output channel divergence, $D^{\inf}(\cN\|\cM) = \inf_{\rho,\sigma} D(\cN(\rho)\|\cM(\sigma))$ is equal to zero if and only if there exist input states $\rho$ and $\sigma$ such that the corresponding output states $\cN(\rho)$ and $\cM(\sigma)$ are identical. While this can occur in certain special cases, it is not generic. Here we can list a few examples:
\begin{enumerate}
\item In recovering the quantum Stein’s lemma between states, we take both channels as replacer channels. In this case, $D^{\inf}$ reduces to the quantum relative entropy between their fixed output states, which is zero if and only if those states are identical. 
\item In many practically relevant scenarios—such as determining whether a device produces a resource state or a junk (free) state (e.g., channels that prepare a fixed entangled state versus entanglement-breaking channels, or channels that prepare a fixed coherent state versus complete dephasing channels)—the outputs cannot coincide, as the set of resource states have no overlap with the set of free states by definition.
\item Moreover, in the context of entropy accumulation, $D^{\inf}$ reduces to the conditional entropy $\sup_{\omega} H(S_i|C_i)_{\cN(\omega)}$ which is generically nonzero. \item Finally, for discriminating between two quantum erasure channels—a noise model commonly used in photonic and trapped-ion systems— $D^{\inf}$ equals zero only if the two channels are identical.
\end{enumerate} 
}

{Moreover, the minimum output channel divergence is generally non-additive~\cite{fang2025efficient}, which necessitates taking the regularized limit in the Stein's exponent $D^{\inf,\reg}(\cN\|\cM)$. While regularization typically complicates estimation, the fact that the optimal exponent is achieved by non-adaptive strategies leads to a significant simplification: it admits explicit convergence bounds, even in the presence of regularization. This exhibits a unique feature in quantum information theory and enables efficient computation of the exponent via a polynomial-time algorithm. Specifically, as the sets $\sA_n'$ and $\sB_n'$ from non-adaptive strategies fall within the framework of the generalized AEP in~\cite{fang2024generalized}, the Stein's exponent $D^{\inf,\reg}(\cN\|\cM)$ for $\cN \in \CPTP(A:B)$ and $\cM \in \CP(A:B)$ can be approximated within an additive error $\delta$ by a quantum relative entropy program of size $O((l+1)^k)$, where $k = \max\{|A|^2, |B|^2\}$ is given by the channel dimensions, and $l = \lceil \frac{8|B|^2}{\delta} \log \frac{|B|^2}{\delta} \rceil$ relates to the expected accuracy. Further details on this computational aspect can be found in the accompanying paper~\cite{fang2025efficient}.
In sharp contrast, for best-case channel discrimination, the optimal exponent is not even known to be computable by any means, regardless of computational resources.}

\section{Relative entropy accumulation}
\label{sec: Application 3: entropy accumulation theorem}

The entropy accumulation theorem~\cite{dupuis2020entropy,metger2022generalised} is a technique to find bounds on the operationally relevant uncertainty (entropy) present in the outputs of a sequential process as a sum of the worst case uncertainties (entropies) of each step. It has been widely used in quantum cryptography~\cite{arnon2018practical,liu2018device,george2022finite}. This naturally raises the question of whether such a statement can be generalized to divergences between arbitrary sequential processes of channels, rather than being restricted only to entropies. This was first asked as an open question in~\cite{metger2022generalised} for the max-relative entropy.

More specifically, the operational setting for relative entropy accumulation is depicted in Figure~\ref{fig:eat setting}. Consider two states $\rho_{A_1}$ and $\sigma_{A_1}$ and quantum channels $\cN_i \in \CPTP(A_i\!:\!A_{i+1} B_i)$ and $\cM_i \in \CP(A_i\!:\!A_{i+1} B_i)$ that are applied sequentially from $i=1$ to $i=n$ and generating the systems $B_i$. The systems $A_i$ should be seen as an internal memory system that we do not control. The key question in the relative entropy accumulation asks: \emph{Can we bound the operationally relevant divergence between the obtained states as the sum of the contributions of each step?}

\begin{figure}[H]
    \centering
    \includegraphics[width=0.6\linewidth]{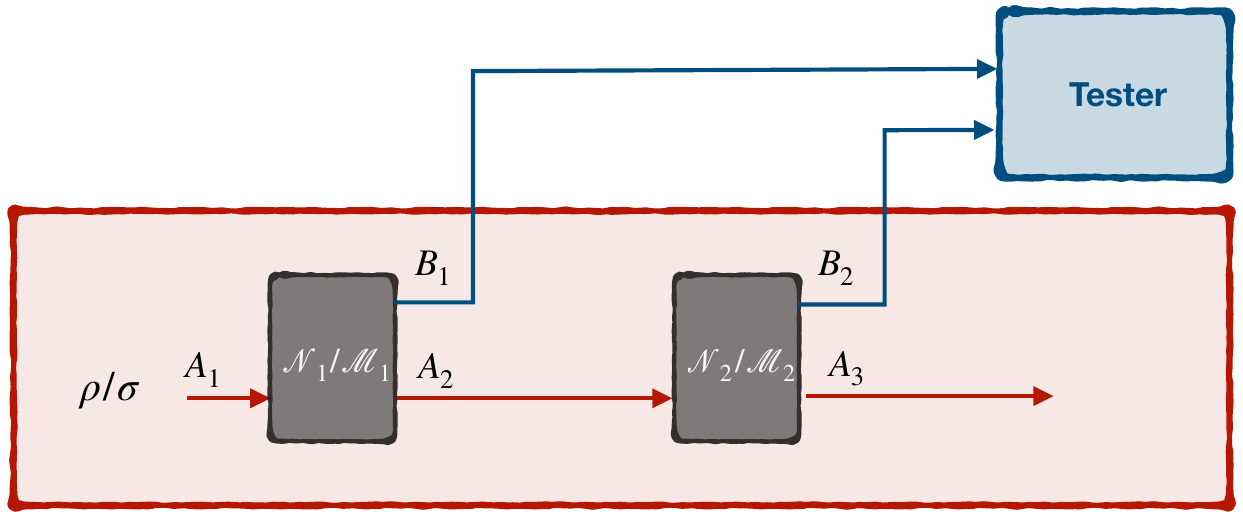}
    \caption{Illustration of the setting of the relative entropy accumulation.}
    \label{fig:eat setting}
\end{figure}

Similar to the adversarial channel discrimination in the main text, we denote the Stinespring dilations of $\cN_i$ and $\cM_i$ as $U_i$ and $V_i$, respectively, with the environmental system denoting as $T_i$. Denote also the corresponding channel as $\cU_i(\cdot)=U_i(\cdot)U_i^\dagger$ and $\cV_i(\cdot) = V_i(\cdot)V_i^\dagger$. Note that $U_i$ is an isometry because $\cN_i$ is trace-preserving, but this is not necessarily the case for $V_i$. Then the setting in Figure~\ref{fig:eat setting} is equivalent to the diagram in Figure~\ref{fig:eat setting dilation}. 

\begin{figure}[h]
    \centering
    \includegraphics[width=0.6\linewidth]{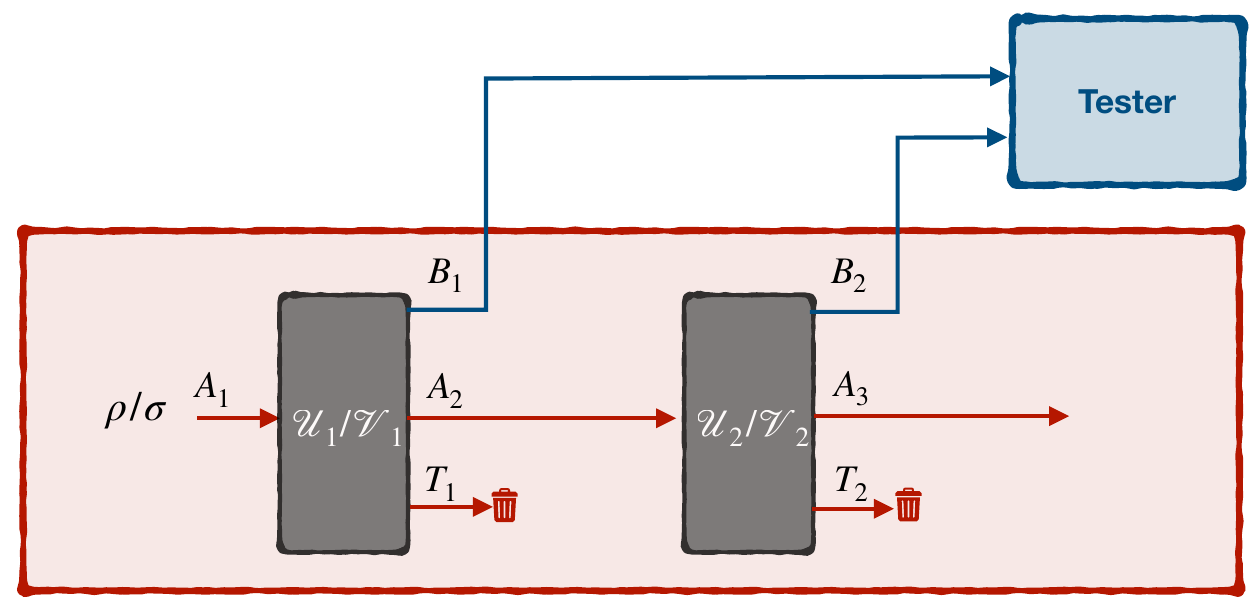}
    \caption{Illustration of the setting of the relative entropy accumulation with Stinespring dilation.}
    \label{fig:eat setting dilation}
\end{figure}

By extending the adversarial channel discrimination framework to allow different channels $\cN_i$ and $\cM_i$ at each round, we can unify the relative entropy accumulation within this broader framework. More specifically, using a similar notation in the main text, we denote the Stinespring dilation of $\cN_i$ and $\cM_i$ as $U_i$ and $V_i$, respectively. Denote also the corresponding channel as $\cU_i(\cdot)=U_i(\cdot)U_i^\dagger$ and $\cV_i(\cdot) = V_i(\cdot)V_i^\dagger$. An illustrative figure is given in Figure~\ref{fig:extended channel discrimination}. Then, taking a particular choice,
\begin{align}
    \dim R_i & = 1,\\
    E_i & = A_{i+1}T_i\\
    \cP^1 & = \rho_{A_1},\\
    \cQ^1 & = \sigma_{A_1},\\
    \cP^i & = \cQ^i = \tr_{T_{i-1}}, \quad \forall i \geq 2,
\end{align}
the adversarial discrimination framework in Figure~\ref{fig:extended channel discrimination} reduces to the relative entropy accumulation in Figure~\ref{fig:eat setting dilation}.

\begin{figure}[H]
    \centering
    \includegraphics[width=0.62\linewidth]{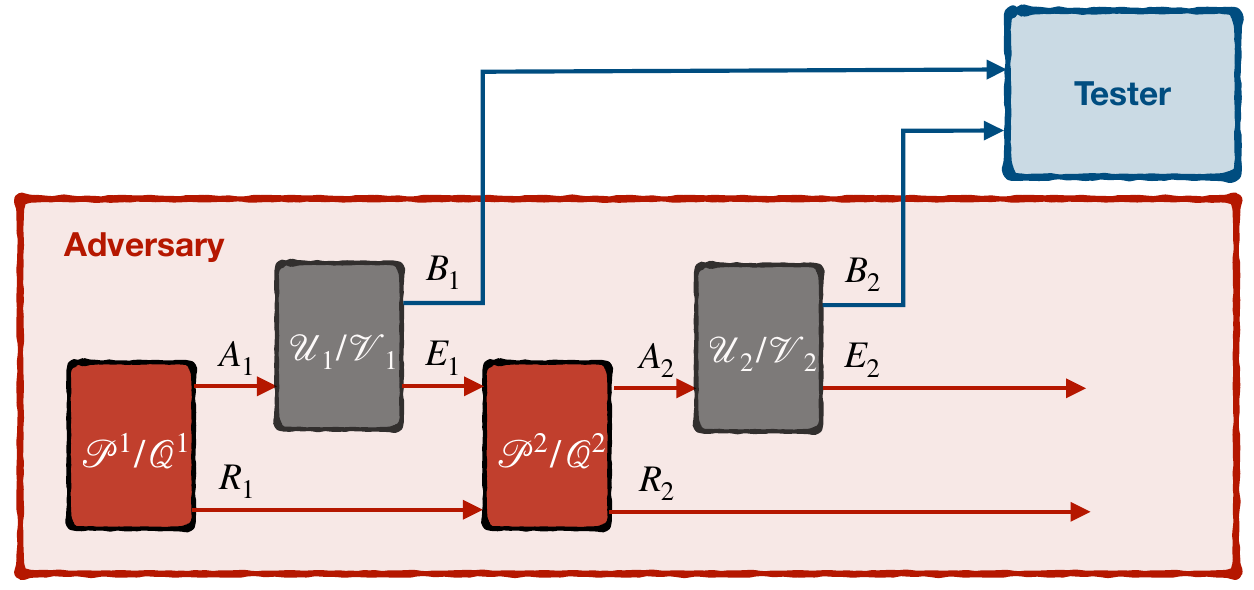}
    \caption{Illustration of an extended setting of adversarial quantum channel discrimination. 
    }
    \label{fig:extended channel discrimination}
\end{figure}

This new perspective enables us to establish a relative entropy accumulation theorem for $D_{\Hypo, \ve}$ (a smoothed form of the min-relative entropy), giving an answer to the dual formulation of the open question in~\cite{metger2022generalised}. Specifically, we generalize the converse part of the adversarial quantum Stein's lemma in two ways: we allow the channels applied at different steps to vary, and we compute explicit finite-size bounds. 


We state the relative entropy accumulation theorem in this more general setting depicted in Figure~\ref{fig:extended channel discrimination}. From the above discussion, it is clear that the version stated in the main text is a special case.
\begin{theorem}
\label{thm:reat}(Relative entropy accumulation theorem.)
Consider two sequences of maps $\cU_{i} \in \CPTP(A_i\!:\!B_i E_{i})$ and $\cV_i \in \CP(A_i\!:\!B_i E_{i})$ for $i \in \{1,\dots,n\}$ and let $\cP^i, \cQ^{i} \in \CPTP(E_{i-1}R_{i-1}, A_iR_{i})$ as illustrated in Figure~\ref{fig:extended channel discrimination}.
Let $\ve \in (0,1)$ and assume for all $i\in \{1, \dots, n\}$ and $m \geq 2$,
\begin{align}
\label{eq: assumptionC N M}
     \forall \alpha \in [1/2,1]: \; D_{\Petz,3/2}(\rho^{(\alpha)}_m \| \sigma^{(\alpha)}_m) \leq \frac{C}{4}m \; \text{ and } \; \forall \sigma \in \density(A_i) : \; \log \tr(\cV_i(\sigma)) \leq \frac{C}{4}, 
\end{align}
where 
$\rho^{(\alpha)}_m, \sigma^{(\alpha)}_m$ are outputs of the channels $(\tr_{E_{i}} \circ \cU_i)^{\otimes m},(\tr_{E_{i}} \circ \cV_i)^{\otimes m}$ (respectively) and achieve the minimum $D_{\Petz,\alpha}(\rho_m^{(\alpha)} \| \sigma_m^{(\alpha)}) = D^{\inf}_{\Petz,\alpha}( (\tr_{E_{i}} \circ \cU_i)^{\otimes m} \| (\tr_{E_{i}} \circ \cV_i)^{\otimes m})$. 
Then letting $\rho_n = \tr_{R_nE_n} \prod_{i=1}^n \cU_{i} \circ \cP^{i}$ and $\sigma_n = \tr_{R_n E_n} \prod_{i=1}^n \cV_{i} \circ \cQ^{i}$, we have
\begin{align}
\label{eq: reat eq}
D_{\Hypo, \ve} & \left(\rho_n  \bigg\| \sigma_n\right)  \geq \sum_{i=1}^n D^{\inf,\reg}(\tr_{E_{i}} \circ \cU_i \| \tr_{E_{i}} \circ \cV_i) - C' n^{2/3} \log n \log^{1/3} \frac{1}{\ve}, 
\end{align}
where $C'$ is a constant that only depends on $d = \max_{i} \dim B_i$ and $C$.
\end{theorem}

\begin{proof}
The proof below is basically the same as the converse part of the adversarial quantum Stein's lemma in the main text. We start as usual by bounding the hypothesis testing relative entropy with a R\'enyi divergence of order $\alpha \in (0,1)$:
\begin{align}
D_{\Hypo, \ve}\left(\rho_n \bigg\| \sigma_n \right) 
\geq D_{\Sand,\alpha}\left(\rho_n \bigg\| \sigma_n \right) + \frac{\alpha}{\alpha-1}\log(1/\ve).
\end{align}
We can now introduce $\rho'_n = \tr_{R_n} \circ \cP^n \circ \prod_{i=1}^{n-1} \cU_i \circ \cP^i$ and $\sigma'_n = \tr_{R_n} \circ \cQ^n \circ \prod_{i=1}^{n-1} \cV_i \circ \cQ^i$ the joint states before the $n$-th use the channel $\cU_{i}/\cV_{i}$. We have $\rho_n = (\tr_{E_{n}} \circ \cU_{n})(\rho'_n)$ We use the chain rule for sandwiched relative entropy to bound 
\begin{align}
D_{\Sand,\alpha}\left( (\tr_{E_n} \circ \cU_{n})(\rho'_n) \bigg\| (\tr_{E_n} \circ \cV_n)(\sigma'_n)\right) &\geq D_{\Sand,\alpha}\left( \tr_{A_n} \rho'_n \bigg\| \tr_{A_n} \sigma'_n \right) + D^{\inf, \infty}_{\Sand,\alpha}(\tr_{E_n} \circ \cU_{n} \| \tr_{E_n} \circ \cV_n).
\end{align}
Now note that $\tr_{A_n} \rho'_n = \rho_{n-1}$ where $\rho_{n-1} = \tr_{R_{n-1}E_{n-1}} \prod_{i=1}^{n-1} \cU_{i} \circ \cP^{i}$ and similarly for $\sigma$. As a result, applying the chain rule $n-1$ times, we get
\begin{align}
D_{\Sand,\alpha}\left( \rho_n \bigg\| \sigma_n \right)
&\geq \sum_{i=1}^n D^{\inf,\infty}_{\Sand,\alpha}(\tr_{E_{i}} \circ \cU_i \| \tr_{E_{i}} \circ \cV_i) \\
&= \sum_{i=1}^n D^{\inf,\infty}_{\Meas,\alpha}(\tr_{E_{i}} \circ \cU_i \| \tr_{E_{i}} \circ \cV_i) \\
&\geq \sum_{i=1}^n \frac{1}{m} D^{\inf}_{\Meas,\alpha}((\tr_{E_{i}} \circ \cU_i)^{\otimes m} \| (\tr_{E_{i}} \circ \cV_i)^{\otimes m})
\end{align}
where $m \geq 2$ and we used~\cite[Lemma 28]{fang2024generalized} for the equality and the superadditivity of $D_{\Meas, \alpha}$ in~\cite[Lemma 21]{fang2024generalized} for the last inequality. Note that both of these results were applied for the family of states $\sA_m'' = (\tr_{E_{i}} \circ \cU_i)^{\otimes m}(\density)$ and $\sB_m'' = (\tr_{E_{i}} \circ \cV_i)^{\otimes m}(\density)$ which satisfies~\cite[Assumption 24]{fang2024generalized}, as shown in the proof of the adversarial quantum Stein's lemma.

Observe that assumption~\eqref{eq: assumptionC N M} implies assumption ($*$) in \cite[Lemma 30]{fang2024generalized} as $\log \tr(\cV_i^{\otimes m}(\sigma)) \leq m \log \lambda_{\max}(\cV_i^{\dagger}(I)) \leq \frac{C}{4} m$. As a result, \cite[Lemma 30]{fang2024generalized} gives
\begin{align}
    \frac{1}{m} D^{\inf}_{\Meas,\alpha}&((\tr_{E_{i}} \circ \cU_i)^{\otimes m} \| (\tr_{E_{i}} \circ \cV_i)^{\otimes m})\notag \\
    &\geq D^{\inf, \infty}(\tr_{E_{i}} \circ \cU_i \| \tr_{E_{i}} \circ \cV_i) - (1-\alpha) (2+C)^2 m - \frac{2(d^2+d) \log(m+d)}{m},
\end{align}
for $1-\frac{1}{(2+C)m} < \alpha < 1$. Let us now choose $1-\alpha = \frac{8d^2 \log m}{(2+C)^2 m^2}$ and assume that $m \geq \max\left(d,\left(\frac{16 d^2}{2+C}\right)^2\right)$ so that the condition $\alpha \geq 1 - \frac{1}{(2+C)m}$ is satisfied and $\log(m+d) \leq 2 \log m$. 

Putting everything together, we get
\begin{align}
D_{\Hypo, \ve}  \left(\rho_n \bigg\| \sigma_n \right) \geq \sum_{i=1}^n D^{\inf,\reg}(\tr_{E_{i}} \circ \cU_{i} \| \tr_{E_{i}} \circ \cV_{i}) - n \frac{16d^2 \log m}{m} - \frac{(2+C)^2 m^2}{8d^2\log m}\log \frac{1}{\ve}.
\end{align}
We now choose $m = \left(\frac{64 d^4 n }{(2+C)^2 \log \frac{1}{\ve} }\right)^{1/3} $.  With this choice
\begin{align}
D_{\Hypo, \ve}\left(\rho_n \bigg\| \sigma_n \right)
&\geq \sum_{i=1}^n D^{\inf,\reg}(\tr_{E_{i}} \circ \cU_{i} \| \tr_{E_{i}} \circ \cV_{i}) - C' n^{2/3} \log n \log^{1/3} \frac{1}{\ve},
\end{align}
for a constant $C'$ that only depends on $C$ and $d$.
\end{proof}

\begin{figure}[H]
    \centering
    \includegraphics[width=0.6\linewidth]{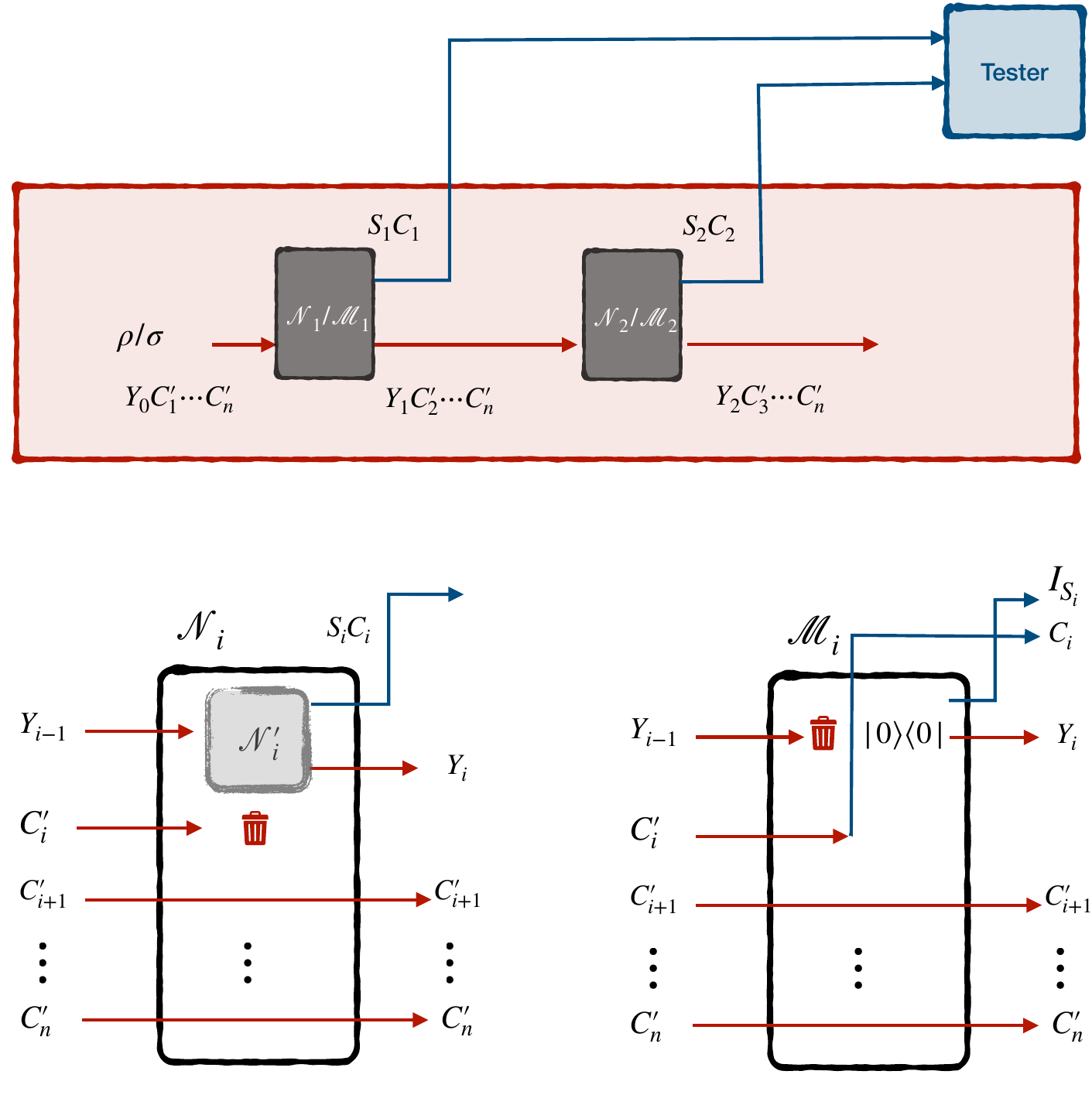}
    \caption{Illustration of the setting of the entropy accumulation.}
    \label{fig:setting EAT}
\end{figure}

As a corollary, we get a weak form of the entropy accumulation statement for the smoothed max-entropy obtained in~\cite{metger2022generalised}. We recall that the smoothed max-entropy is defined as
\begin{align}
    H_{\max}^{\ve}(B|C)_{\rho} = \log \inf_{\substack{\tilde{\rho}_{BC} \in \PSD(BC) \\ \tr(\tilde{\rho}) \leq 1 \\ P(\rho, \tilde{\rho}) \leq \ve}} \sup_{\sigma_C \in \density(C)} \left\| \tilde{\rho}_{BC}^{\frac{1}{2}} \id_{B} \otimes \sigma_C^{\frac{1}{2}} \right\|_1^2,
\end{align}
where $P(\rho,\sigma) = \sqrt{1-F(\rho,\sigma)^2}$ is the purifided distance with $F(\rho,\sigma) := \|\sqrt{\rho}\sqrt{\sigma}\|_1 + \sqrt{(1-\tr\rho)(1-\tr\sigma)}$.

\begin{corollary}($H_{\max}$-entropy accumulation.)
    Let $\cN_i' \in \CPTP(Y_{i-1}\!:\!S_i C_i Y_{i})$ be quantum channels and $\rho_{Y_0} \in \density(Y_0)$ be a quantum state. Define the state $\rho_{S_1 \dots S_n C_1 \dots C_n Y_n} = \prod_{i=1}^n \cN_i' (\rho_{Y_0})$. We have for $\ve \in [0,1/2]$,
    \begin{align}
        H_{\max}^{\ve}(S_1 \dots S_n | C_1 \dots C_n)_{\rho} \leq \sum_{i=1}^n \sup_{\omega \in \density(Y_{i-1})} H(S_i | C_i)_{\cN'(\omega)} + K n^{2/3} \log n \log^{1/3} \frac{1}{\ve},
    \end{align}
    where $K$ is a constant only depending on $\max_{i} \dim S_i C_i$.
\end{corollary}

\begin{proof}
We apply the relative entropy accumulation theorem as stated in the main text with the following replacements: letting $C'_i$ be isomorphic to $C_i$, we set $A_{i+1} \leftarrow Y_{i} C'_{i+1} \dots C'_n$ for $i=0$ to $n-1$, $A_{n+1} \leftarrow Y_n$,  $B_i \leftarrow S_i C_i$ for $i=1$ to $n$, $\cN_i \leftarrow \cN_{i}' \otimes \tr_{C'_{i}} \otimes \operatorname{id}_{C'_{i+1} \dots C'_n}$ (here, $\operatorname{id}$ refers to the identity map) and finally $\cM_i$ is defined as $\cM_i(X_{Y_{i-1}C'_i \dots C'_n}) = \id_{S_i} \otimes X_{C_i C'_{i+1} \dots C'_{n}} \otimes \proj{0}_{Y_i}$, where $\proj{0}_{Y_i}$ is an arbitrary fixed state in $\density(Y_i)$. An illustrative figure is given in Figure~\ref{fig:setting EAT}.

Theorem~\ref{thm:reat} gives for the left hand side of~\eqref{eq: reat eq} and any $\sigma \in \density(C_1' \dots C_n')$:
\begin{align}
D_{\Hypo, \ve}&\left(\tr_{A_{n+1}} \circ \prod_{i=1}^n \cN_i (\rho_{Y_0} \otimes \sigma_{C'_1 \dots C'_n}) \bigg\| \tr_{A_{n+1}} \circ \prod_{i=1}^n \cM_i (\proj{0}_{Y_0} \otimes \sigma_{C'_1 \dots C'_n})\right)\notag \\
    &= D_{\Hypo, \ve}(\rho_{S_1 \dots S_n C_1 \dots C_n} \| \id_{S_1 \dots S_n} \otimes \sigma_{C_1 \dots C_n}).
\end{align}
As the right hand side of~\eqref{eq: reat eq} does not depend on $\sigma$, we can take an infimum over $\sigma$ and use~Proposition~\ref{prop: DH Hmax} to get the following
\begin{align}
\inf_{\sigma} D_{\Hypo, \ve}&\left(\tr_{A_{n+1}} \circ \prod_{i=1}^n \cN_i (\rho_{Y_0} \otimes \sigma_{C'_1 \dots C'_n}) \bigg\| \tr_{A_{n+1}} \circ \prod_{i=1}^n \cM_i (\proj{0}_{Y_0} \otimes \sigma_{C'_1 \dots C'_n})\right) \notag\\
    &\leq - H_{\max}^{\sqrt{2\ve}}(S_1 \dots S_n | C_1 \dots C_n)_{\rho}.
\end{align}

On the right hand side of~\eqref{eq: reat eq}, we have terms of the form $D^{\inf,\reg}(\tr_{A_{i+1}} \circ \cN_i \| \tr_{A_{i+1}} \circ \cM_i)$. Note that for $\omega \in \density(Y_{i-1}C'_{i} \dots C'_n)$, we have $(\tr_{A_{i+1}} \circ \cN_i)(\omega) = \tr_{Y_i}\cN_i'(\omega_{Y_{i-1}})$ and $(\tr_{A_{i+1}} \circ \cM_i)(\omega) = \id_{S_i} \otimes \omega_{C_i}$. As a result, $D^{\inf}(\tr_{A_{i+1}} \circ \cN_i \| \tr_{A_{i+1}} \circ \cM_i) = \inf_{\omega \in \density(Y_{i-1})} - H(S_i | C_i)_{\cN_i'(\omega)}$, where we used the fact that $-H(B|C)_{\rho} = \inf_{\sigma \in \density(C)} D(\rho_{BC} \| \id_{B} \otimes \sigma_{C})$. We now need to evaluate the regularization:
\begin{align}
    D^{\inf,\reg}&(\tr_{A_{i+1}} \circ \cN_i \| \tr_{A_{i+1}} \circ \cM_i) \notag \\
    &= \inf_{m \geq 1} \frac{1}{m}  \inf_{\substack{\omega \in \density((Y_{i-1})^{\otimes m}) \\ \sigma \in \density((C_{i} \dots C_n)^{\otimes m})}} D((\tr_{Y_i}\circ\cN_i')^{\otimes m}(\omega) \| \id_{S_{i,1} \dots S_{1,m}} \otimes \sigma_{C_{i,1} \dots C_{i,m}}) \\
    &= - \sup_{m \geq 1} \sup_{\omega} \frac{1}{m} H(S_{i,1} \dots S_{i,m} | C_{i,1} \dots C_{i,m})_{\cN_i'^{\otimes m}(\omega)} \\
    &= - \sup_{m \geq 1} \sup_{\omega} \frac{1}{m} \sum_{j=1}^m H(S_{i,j} | C_{i,1} \dots C_{i,m} S_{i,1} \dots S_{i,j-1})_{\cN_i'^{\otimes m}(\omega)} \\
    &\geq - \sup_{m \geq 1} \sup_{\omega} \frac{1}{m} \sum_{j=1}^m H(S_{i,j} | C_{i,j} )_{\cN_i'(\omega_{Y_{i,j}})},
\end{align}
where we used the chain rule and then strong subadditivity of the von Neumann entropy.
Now note that each term $H(S_{i,j} | C_{i,j} )_{\cN_i'(\omega_{Y_{i,j}})} \leq \sup_{\omega} H(S_i | C_i)_{\cN'_i(\omega)}$. As a result, $D^{\inf,\reg}(\tr_{A_{i+1}} \circ \cN_i \| \tr_{A_{i+1}} \circ \cM_i) \geq - \sup_{\omega \in \density(Y_{i-1})} H(S_i | C_i)_{\cN_i'(\omega)}$
which shows that in this case, the regularization is not needed. In order to apply Theorem~\ref{thm:reat}, we need to check condition~\eqref{eq: assumptionC N M}. First, we have for any $\omega \in \density(Y_i)$, $\tr(\cM_i(\omega)) = \dim S_i$. In addition, for $\alpha \in [1/2,1]$, we have
\begin{align}
    D^{\inf}_{\Petz,\alpha}&((\tr_{A_{i+1}} \circ \cN_i)^{\otimes m} \| (\tr_{A_{i+1}} \circ \cM_i)^{\otimes m}) \notag\\
    &=\inf_{\substack{\omega \in \density((Y_{i-1})^{\otimes m}) \\ \sigma \in \density((C_{i} \dots C_n)^{\otimes m})}} D_{\Petz, \alpha}((\tr_{Y_i}\cN_i')^{\otimes m}(\omega_{Y_{i-1,1} \dots Y_{i-1,m}}) \| \id_{S_{i,1} \dots S_{1,m}} \otimes \sigma_{C_{i,1} \dots C_{i,m}}) \label{eq:expression-dalpha-opt} \\
    &= - \sup_{\omega} H^{\uparrow}_{\Petz, \alpha}(S_{i,1} \dots S_{i,m} | C_{i,1} \dots C_{i,m})_{\cN_i'^{\otimes m}(\omega)},
\end{align}
using the notation $H^{\uparrow}_{\Petz, \alpha}(D| E)_{\rho} = -\inf_{\sigma \in \density(E)}D_{\Petz,\alpha}(\rho_{DE} \| \id_{D} \otimes \sigma_E)$ from~\cite{tomamichel2014relating}. It is shown in~\cite[Lemma 1]{tomamichel2014relating}, that an explicit choice of $\sigma$ achieves this infimum namely $\sigma^{(\alpha)}_{E} = \frac{(\tr_D \rho_{DE}^{\alpha})^{1/\alpha}}{\tr(\tr_D \rho_{DE}^{\alpha})^{1/\alpha}}$. But in~Lemma~\ref{lem: assumptionC for cond entr}, we showed that for this choice $D_{\Petz,3/2}(\rho_{DE} \| \id_{D} \otimes \sigma^{(\alpha)}_{E}) \leq 4 \log \dim D$. Applying this to the state $\rho = (\tr_{A_{i+1}} \circ \cN_i)^{\otimes m}(\omega)$, an optimal choice for $\sigma$ in~\eqref{eq:expression-dalpha-opt} is given by $\sigma_{m}^{(\alpha)} = \frac{(\tr_{B_{i}^{\otimes m}} \rho^{\alpha})^{1/\alpha}}{\tr(\tr_{B_{i}^{\otimes m}} \rho^{\alpha})^{1/\alpha}}$. We get that for any $\omega \in \density(R_{i-1}^{\otimes m})$,
\begin{align}
    D_{\Petz,3/2}((\tr_{A_{i+1}} \circ \cN_i)^{\otimes m}(\omega) \| \id_{S_i}^{\otimes m} \otimes \sigma^{(\alpha)}_{m}) \leq 4 m \log \dim S_i.
\end{align}
This means that choosing $C = 16 \max_{i} \log \dim S_i$ satisfies condition~\eqref{eq: assumptionC N M}.
\end{proof}

\smallskip

Note that the second order term we achieve with our proof technique is worse than the one achieved in~\cite{dupuis2020entropy} and~\cite{metger2022generalised}. In addition, the statement of~\cite{metger2022generalised} is stronger in that it also includes conditioning on the system $Y_i$ provided a non-signalling assumption is satisfied. Nevertheless, we believe that our new proof technique, which is more naturally adapted to find upper bounds on the max-entropy $H_{\max}$ (as opposed to the techniques of~\cite{metger2022generalised} which naturally apply to the min-entropy $H_{\min}$) could lead to insights and improvements for the applications of entropy accumulation. However, this is outside the scope of this paper and we leave it for future work.

\section{Useful properties}

\begin{lemma}
\label{lem: assumptionC for cond entr}
    Let $\alpha \in [1/2,1]$, $\rho_{AB} \in \density(AB)$ and $\sigma^{(\alpha)}_B = \frac{\left(\tr_A\rho_{AB}^{\alpha}\right)^{\frac{1}{\alpha}}}{Z}$, where $Z = \tr \left(\left(\tr_A\rho_{AB}^{\alpha}\right)^{\frac{1}{\alpha}}\right)$. Let $d_A = \dim A$. Then it holds that
    \begin{align}
        D_{\Petz,\frac{3}{2}}(\rho_{AB} \| \id_{A} \otimes \sigma_{B}^{(\alpha)}) \leq 4\log d_A.
    \end{align}
\end{lemma}
Note that the choice of the parameter $\frac{3}{2}$ is not arbitrary and this lemma does not hold for higher values. See the discussion in the proof of ~\cite[Corollary III.5]{dupuis2019entropySecondOrder}.

\smallskip
\begin{proof}
    We have
    \begin{align}
        \tr\left(\rho_{AB}^{3/2} (\id_{A} \otimes \sigma_{B}^{(\alpha)})^{-1/2}\right)
        &= Z^{1/2} \tr\left(\rho_{AB}^{3/2} \id_{A} \otimes (\tr_{A} \rho_{AB}^{\alpha})^{-\frac{1}{2\alpha}}\right).
    \end{align}
    We start by showing that $Z \leq d_{A}$. In fact, we use the operator Jensen inequality for the operator concave function $x \mapsto x^{\alpha}$ as follows:
    \begin{align}
        \tr_{A}(\rho_{AB}^{\alpha}) 
        &= d_A \sum_{a} \frac{\bra{a}}{\sqrt{d_A}} \rho_{AB}^{\alpha} \frac{\ket{a}}{\sqrt{d_A}} \\
        &\leq d_A \left(\sum_{a} \frac{\bra{a}}{\sqrt{d_A}} \rho_{AB} \frac{\ket{a}}{\sqrt{d_A}}\right)^{\alpha} \\
        &= d_A^{1-\alpha} \rho_B^{\alpha}.
    \end{align}
    As a result, we get
    \begin{align}
        Z &\leq \tr( (d_A^{1-\alpha} \rho_B^{\alpha})^{\frac{1}{\alpha}}) = d_{A}^{\frac{1-\alpha}{\alpha}} \leq d_A.
    \end{align}
    For the inequality, we used the fact that the function $x \mapsto x^{1/\alpha}$ is monotone and continuous and thus $X \mapsto \tr(X^{\frac{1}{\alpha}})$ is monotone (see e.g.,~\cite[Section 2.2]{carlen2010trace}).

    Now let us consider 
    \begin{align}
    \tr\left(\rho_{AB}^{3/2} \id_{A} \otimes (\tr_{A} \rho_{AB}^{\alpha})^{\frac{1}{\alpha}}\right) 
    &\leq \tr\left(\rho_{AB} \id_{A} \otimes (\tr_{A} \rho_{AB})^{-\frac{1}{2\alpha}}\right) \\
    &\leq \tr(\rho_{AB} \rho_{B}^{-1}) \\
    &\leq d_A,
    \end{align}
    where we used for the first inequality the fact that $\rho_{AB}^{3/2} \leq \rho_{AB}$, $\rho_{AB}^{\alpha} \geq \rho_{AB}$ and the operator anti-monotonicity of the function $x \mapsto x^{-\frac{1}{2\alpha}}$, and for the second inequality the fact that $\rho_{B}^{\frac{1}{2\alpha}} \geq \rho_{B}$. As a result,
    \begin{align}
        D_{\Petz,\frac{3}{2}}(\rho_{AB} \| \id_{A} \otimes \sigma_B^{(\alpha)}) \leq 2 \log (d_A^{1/2} d_A) \leq 4 \log d_A.
    \end{align}
\end{proof}

\begin{proposition}
\label{prop: DH Hmax}
    For $0 \leq \ve < \frac{1}{2}$, we have
    \begin{align}
    H_{\max}^{\sqrt{2\ve}}(B|C)_{\rho} \leq - \inf_{\sigma_{C} \in \density(C)} D_{\Hypo, \ve}(\rho_{BC} \| \id_{B} \otimes \sigma_C).
\end{align}
\end{proposition}
\begin{proof}
Using \cite[Lemma 31]{fang2024generalized} with $\sA = \{\rho_{BC}\}$ and $\sB = \{\id_{B} \otimes \sigma_C : \sigma_{C} \in \density(C)\}$, we can write
\begin{align}
    -\inf_{\sigma_{C} \in \density(C)} D_{\Hypo, \ve}(\rho_{BC} \| \id_{B} \otimes \sigma_C) = \log \inf_{0 \leq M \leq I} \left\{ \sup_{\sigma_{C} \in \density(C)} \{\tr(M_{BC} I_B \otimes \sigma_C) : \tr(M_{BC} \rho_{BC}) \geq 1- \ve \right\}.
\end{align}
Let $M_{BC}$ be such that $\tr(M_{BC}\rho_{BC}) \geq 1-\ve$, define $\tilde{\rho}_{BC} = \sqrt{M_{BC}} \rho_{BC} \sqrt{M_{BC}}$. Then, by the gentle measurement lemma, we have $P(\rho, \tilde{\rho}) \leq \sqrt{2\ve}$ (see e.g.,~\cite[Lemma A.3]{dupuis2014generalized}). Then, using Lemma~\ref{lem:bound on fidelity}, we get
\begin{align}
    \sup_{\sigma_{C}\in \density(C)} \left\| (\sqrt{M_{BC}} \rho_{BC} \sqrt{M_{BC}})^{\frac{1}{2}} \id_B \otimes \sigma_C^{\frac{1}{2}} \right\|_1^2 \leq \sup_{\sigma_C} \tr(M_{BC} I_B \otimes \sigma_C)
\end{align}
As such,
\begin{align}
    -\inf_{\sigma_{C}  \in \density(C)} & D_{\Hypo, \ve}(\rho_{BC} \| \id_{B} \otimes \sigma_C) \notag \\
    &\geq \log \inf_{0 \leq M \leq I}  \sup_{\sigma_{C}} \left\{ \left\| (\sqrt{M_{BC}} \rho_{BC} \sqrt{M_{BC}})^{\frac{1}{2}} \id_B \otimes \sigma_C^{\frac{1}{2}} \right\|_1^2 : \tr(M_{BC} \rho_{BC}) \geq 1- \ve \right\} \\
    &\geq \log \inf_{\substack{\tilde{\rho}_{BC} \in \PSD(BC) \\ \tr(\tilde{\rho}) \leq 1 \\ P(\rho, \tilde{\rho}) \leq \sqrt{2\ve}}} \sup_{\sigma_C \in \density(C)} \left\| \tilde{\rho}_{BC}^{\frac{1}{2}} \id_{B} \otimes \sigma_C^{\frac{1}{2}} \right\|_1^2\\
    &= H_{\max}^{\sqrt{2\ve}}(B|C)_{\rho}.
\end{align}
\end{proof}
\begin{lemma}
\label{lem:bound on fidelity}
    Let $\rho \in \density(A)$ and $\sigma, M \in \PSD(A)$. Then it holds that
    \begin{align}
        \left\| (\sqrt{M}\rho\sqrt{M})^{\frac{1}{2}} \sigma^{\frac{1}{2}} \right\|_1^2 \leq \tr(M \sigma).
    \end{align}
\end{lemma}
\begin{proof}
    This fact is used in~\cite[Proposition 4.2]{dupuis2014generalized}. It uses the semidefinite program for the fidelity $\| \sqrt{\omega} \sqrt{\theta} \|_{1}^2 = \min\{ \tr(Z \theta) : \omega_{AE} \leq Z \otimes I_E, Z \geq 0\}$~\cite[Section 5]{watrous2009semidefinite}, where $\omega_{AE}$ is a purification of $\omega$. Let $\rho_{AE}$ be a purification of $\rho_A$. Then we have $\sqrt{M_A} \rho_{AE} \sqrt{M_A} \leq M_A \otimes I_{E}$ and $\sqrt{M_A} \rho_{AE} \sqrt{M_A}$ is a purification of $\sqrt{M} \rho_A \sqrt{M}$. As a result $Z = M$ is feasible for the semidefinite program above and we get the desired result.
\end{proof}

\end{document}